\documentclass[smallextended]{svjour3}       
\smartqed  
\usepackage{graphicx}
\usepackage{csquotes}
\usepackage[table]{xcolor}
\usepackage[authoryear]{natbib}
\usepackage{amsmath}
\usepackage{amssymb}
\usepackage{mathptmx}
\usepackage{epigraph}
\usepackage[utf8]{inputenc}
\begin{document}

\title{Games in Minkowski Spacetime}
\titlerunning{Decision Making in Minkowski Spacetime}


\author{Ghislain Fourny}


\institute{G. Fourny \at
              ETH Z\"urich \\
              Department of Computer Science \\
              \email{ghislain.fourny@inf.ethz.ch}\\
              ORCID 0000-0001-8740-8866\\
}

\date{April 23, 2020}

\maketitle

\begin{abstract}

This paper contributes a new class of games called spacetime games with perfect information. In spacetime games, the agents make decisions at various positions in Minkowski spacetime. Spacetime games can be seen as the least common denominator of strategic games on the one hand, and dynamic games with perfect information on the other hand. Indeed, strategic games correspond to a configuration with only spacelike-separated decisions (``different rooms''). Dynamic games with perfect information, on the other hand, correspond to timelike-separated decisions (``in turn''). We show how to compute the strategic form and reduced strategic form of spacetime games. As a consequence, many existing solution concepts, such as Nash equilibria, rationalizability, individual rationality, etc, apply naturally to spacetime games. We introduce a canonical injection of the class of spacetime games with perfect information into the class of games in extensive form with imperfect information; we provide a counterexample showing that this is a strict superset. This provides a novel interpretation of a large number of games in extensive form with imperfect information in terms of the theory of special relativity, where non-singleton information sets arise from the finite speed of light. This framework can be a useful tool for reasoning in quantum foundations, where it is important whether decisions such as the choice of a measurement axis or the outcome of a measurement are spacelike- or timelike- separated. We look in particular at the special case of the Einstein-Podolsky-Rosen experiment with four decision points, and model a corresponding spacetime game structure.


\end{abstract}

\section{Introduction}

\subsection{How spacetime was born}

For centuries, our vision of time has been absolute. \cite{Newton1687} designed his laws of motion and his theory of universal gravitation under the assumption that time elapses in the same way everywhere. Concretely, this meant that everybody's clocks, assuming they are synchronized and have ideally high precision, would continue to show the same time to all forever and that there was no ambiguity regarding the order in which events happen: everybody would agree that event A occurs before event B, or that event A and event B occur simultaneously. Furthermore, observations were consistent with the fact that, indeed, the universe seemed to behave in this way.

By the end of the nineteenth century, our knowledge of physics was well covered by Newton's equations for matter, and by Maxwell's equations for electromagnetic waves \citep{Maxwell1865}.

Except that the two theories were not consistent with each other. On the one hand, Newton's laws would require that speeds add up, e.g., if a passenger walks in a train, then a person standing in the train station would see them move with a speed that is the sum of their walking speed and that of the train; something familiar to most people in their everyday life. On the other hand, however, Maxwell's equations required that the speed of light be constant, which led the scientific community to consider the possibility that light moves in an aether medium, in the same way that sound propagates in the air \citep{Michelson1887}.

\cite{Einstein1905} finally solved this problem with a fundamental shift of paradigm: the speed of light in a vacuum does not depend on the observer. This very simple principle had dramatic consequences on our vision of time. First, it led us to completely rethink space and time as forming a unique continuum: spacetime. Second, we had to abandon the idea that all people agree on their clocks. Indeed, in spacetime, the perception of time may differ depending on how fast somebody is moving: the passenger of a spacecraft moving very fast will age less than a person on Earth. Likewise, the size of the spacecraft seen from Earth will also be squeezed.

But more importantly, and this is what will occupy the center of our attention in this paper: \emph{the order in which two events occur does not depend, in some cases, on the observer, while in other cases, it does}. In the former cases, information can be transmitted between the two events. In the latter cases, it cannot. The concept of simultaneity itself makes no sense anymore.

\subsection{Strategic and dynamic games}

A specific kind of event in the field of game theory is decision making: an economic agent who has to make a decision between, say, cooperating or defecting, is located somewhere, and at a certain time. They are making their decision based on the information at their disposal, which is mostly knowledge about the past --- and even \emph{all} the knowledge about the past in the case of perfect information--- as well as any logically inferred consequences thereof \citep{Morgenstern1953}\citep{Kreps1988}.

Game theory envisions two specific kinds of games: games in normal form, also called strategic games, and games in extensive form, also called dynamic games. These games are normally explained under a classical concept of time: for games in normal form, the agents make their decisions simultaneously, while for games in extensive form, they play in turn.

We are now going to go through each one of them, and show how their interpretation can and should be adapted to spacetime and to the fact that there is no such thing as a universal notion of simultaneity.

\subsubsection{Games in normal form}

An example of a game in normal form is shown in Figure \ref{figure-normal-form} \citep{Borel1921} \citep{Neumann1928}. The classical explanation is that Alice and Bob are in separate rooms and do not communicate with each other. They both have to make a decision: cooperate (c) or defect (d). Once they have made a decision, one of four possible payoff distributions (strategy profiles) is identified. In this example, these are the payoffs of the prisoner's dilemma, and the Nash equilibrium \citep{Nash1951} has both defect.

\begin{figure}
\centering\includegraphics[width=0.5\textwidth]{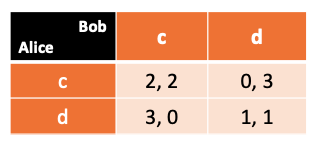}
\caption{An example of a game in normal form. The two players, Alice and Bob, do not communicate, and under Nash assumptions, make their decisions independently.}
\label{figure-normal-form}
\end{figure}

In the context of spacetime and special relativity, where the concept of simultaneity does not make sense, Alice and Bob not communicating with each other can be enforced by sending them far away from each other: for example, Alice on the Earth, and Bob on Mars, as shown in Figure \ref{figure-spacelike}. Light takes at best 4 minutes to travel between the two planets. If Bob makes his decision before the information on Alice's decision reaches him, and vice versa, then we know that no communication was possible. In such a configuration, we say that Alice and Bob's decisions are spacelike-separated (more on this in Section \ref{section-spacetime}). Note that there are many practical use cases that do not require for rational agents to be on separate planets, e.g., high-frequency trading where the durations considered are very small. Considering Mars, however, provides a nice and intuitive illustration of not being able to communicate between two points in spacetime.

\begin{figure}
\centering\includegraphics[width=0.5\textwidth]{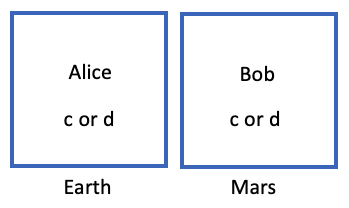}
\caption{An interpretation of the game in normal form that enforces the absence of communication with the properties of spacetime. The two players, Alice and Bob, are located very far away from each other, for example, on the Earth and Mars. Physicists call this spacelike separation.}
\label{figure-spacelike}
\end{figure}

\subsubsection{Games in extensive form (perfect information)}

An example of a game in extensive form is shown in Figure \ref{figure-extensive-form} \citep{Morgenstern1953} \citep{Kuhn1950}. In this game, called the promise game, Alice, the baker first decides to either hand over the bread (c) to Bob or not (d). If Alice handed over the bread to Bob, then Bob can choose to pay (c) or not (d). The subgame perfect equilibrium (which is a special case of a Nash equilibrium) is obtained with a backward induction \citep{Selten1965}, with Bob picking d if he plays, and Alice, knowing this, picking d as well.

More specifically, this is a game in extensive form with \emph{perfect information}, meaning that the players know all the decisions made in previous moves.

\begin{figure}
\centering\includegraphics[width=0.5\textwidth]{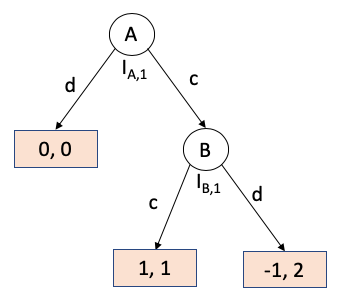}
\caption{An example of a game in extensive form with perfect information. Alice makes her decision first, and then Bob. Bob knows Alice's decision when he plays.}
\label{figure-extensive-form}
\end{figure}

In the context of spacetime and special relativity, Alice and Bob communicating with each other can be enforced by having them located close to each other, but making their decision in turn: for example, both in the same room on Earth, as shown in Figure \ref{figure-timelike}. In such a configuration, Bob knows about Alice's decision when he makes his because a hypothetical beam of light sent by Alice when she decides will have reached him by the time he makes his. In such a configuration, we say that Alice and Bob's decisions are timelike-separated (more on this in Section \ref{section-spacetime}).

\begin{figure}
\centering\includegraphics[width=0.5\textwidth]{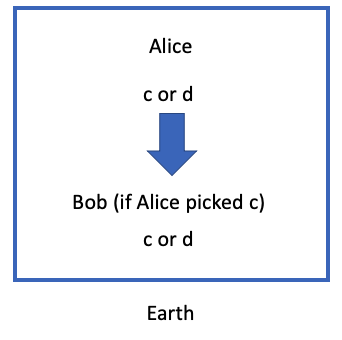}
\caption{An interpretation of the game in extensive form with perfect information that enforces the communication from Alice's decision to Bob with the properties of spacetime. The two players, Alice and Bob, make their decisions one after the other, for example, in the same room on Earth. Physicists call this timelike separation.}
\label{figure-timelike}
\end{figure}

\subsection{Games in extensive form with imperfect information}

As it turns out, and as it is well known, games in normal form and games in extensive form with perfect information are two faces of the same paradigm.

\subsubsection{Strategic form of a dynamic game}

In order to see this, we first observe, as is commonly known by game theorists, that a game in extensive form and with perfect information can be canonically transformed to a game in normal form, called its strategic form. The equivalent of our previous example, the promise game, is shown in Figure \ref{figure-normal-form-of-extensive-form}. The strategic-game equivalent is obtained algorithmically by taking, for each player, the cartesian product of all their decisions to obtain this player's strategy space and filling the payoff matrix accordingly. Duplicate payoffs may be present -- which leaks that some information is actually communicated between the players.

The strategic form of a game in extensive form can, in general, be built with each agent's strategy space being the cartesian product of all their possible choices at every node they are playing. Each such strategy can be interpreted as a kind of ``masterplan'' or as the intentions of the agents, who thought through in advance what they are going to do at every possible node of the game, should it be reached. This strategic form can then be simplified to a so-called reduced strategic form, in which identical strategies can be merged. More details can be found in game theory textbooks such as \citep{Rubinstein1994}.

\begin{figure}
\centering\includegraphics[width=0.5\textwidth]{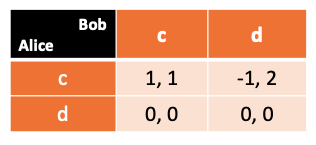}
\caption{The strategic form of the promise game, originally in extensive form with perfect information.}
\label{figure-normal-form-of-extensive-form}
\end{figure}

\subsubsection{Games with simultaneous moves}

This also works the other way round -- and it is this direction that is relevant to our paper. Indeed, a game in normal form can be equivalently expressed as a game in extensive form and with imperfect information. Figure \ref{figure-extensive-form-of-normal-form} shows the prisoner's dilemma, our original game in normal form, expressed in this way. The dotted line between a given set of nodes means that the decision is made across that set of nodes, without knowing at which node exactly one is in this set. In this case, Bob does not know if Alice picked c or d, which is thus equivalent to the original configuration in separate rooms. A game like this is also often referred to as a game in extensive form with simultaneous moves \citep{Dubey1984}.

\begin{figure}
\centering\includegraphics[width=0.5\textwidth]{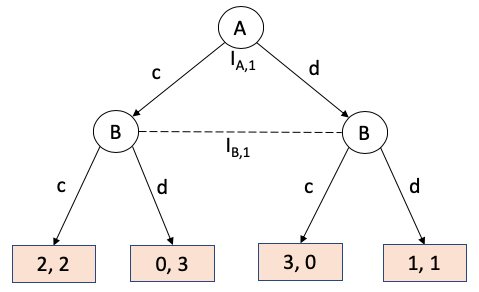}
\caption{The extensive form equivalent of the prisoner's dilemma, originally in normal form. The dotted line materializes the imperfect information.}
\label{figure-extensive-form-of-normal-form}
\end{figure}

\subsubsection{Our contribution}

The above mappings are commonly known by game theorists, although they are occasionally the object of debates and disputes regarding their semantic relevance.

In this paper, we want to adapt and extend the concept of games that have simultaneous moves to a more adequate framework consistent with special relativity. We introduce a new class of games in which the agents make their decisions at various and arbitrary positions in spacetime, and know all the decisions that they can possibly know, i.e., that they can be informed of via speeds lower than the speed of light. We call them spacetime games with perfect information.

We then show that spacetime games with perfect information can be directly mapped to (canonically injected into) games in extensive form with imperfect information, where imperfect information (often represented as dotted lines on the tree representation of the game) is interpreted as the spacelike separation of the decisions of the agents who play the game.

In other words, a large class of games in extensive form with imperfect information is newly interpreted as spacetime games with perfect information in a natural way. This interpretation subsumes the specific interpretation of some games in extensive form with imperfect information as games with simultaneous moves. Formally, as shown in Figure \ref{figure-hierarchy}, we introduce spacetime games as a more fine-grained common denominator of games in normal form on the one hand and games in extensive form with perfect information on the other hand.

\begin{figure}
\centering\includegraphics[width=\textwidth]{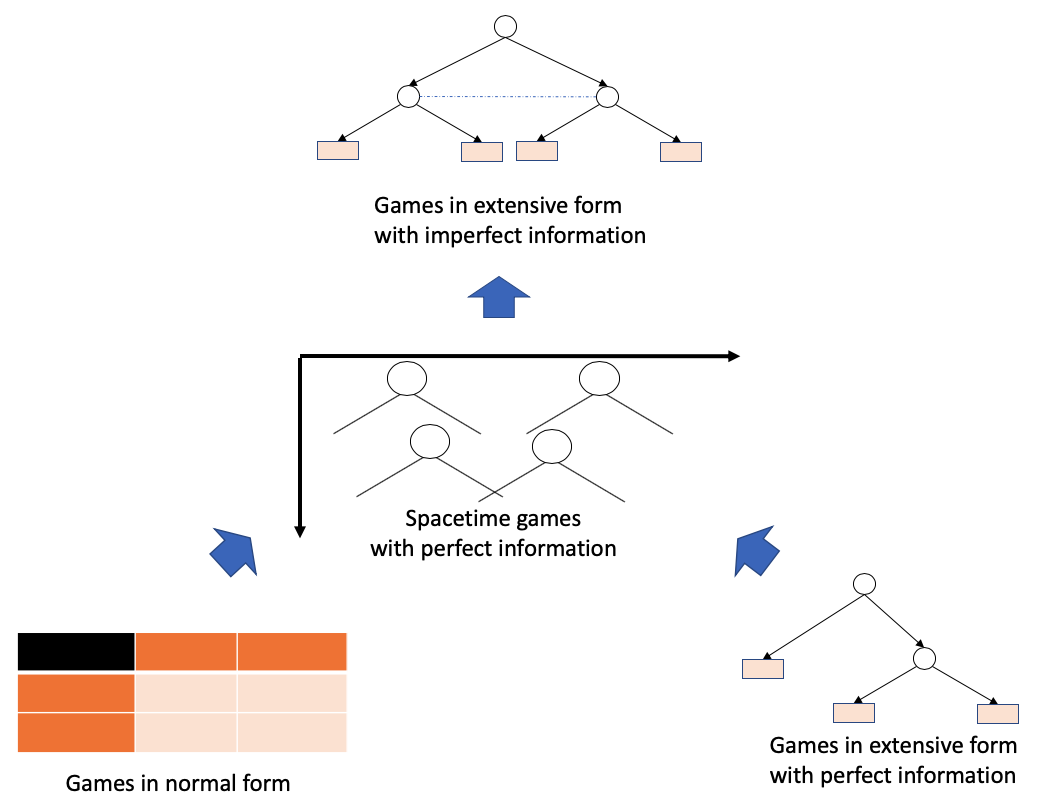}
\caption{The hierarchy of classes of games, enriched spacetime games, acting as a missing link between on the one hand games in normal form and games in extensive form with perfect information, and on the other hand, the most general category of games in extensive form with imperfect information. The big arrows denote class inclusion.}
\label{figure-hierarchy}
\end{figure}

Furthermore, the class of games in extensive form with imperfect information that enjoys a natural interpretation in spacetime according to the above correspondence is a strict subset of the set of all games in extensive form with imperfect information. We provide a counter-example.

\clearpage

\section{Background}

Before we start, we give a short introduction to games with imperfect information as well as on Minkowski spacetime. Readers already familiar with the one or the other can skip the corresponding section, even though we assume it is more likely that the reader will be familiar with game theory but not with special relativity. We will thus attempt to explain special relativity in a way as pedagogical as possible.

\subsection{Games with imperfect information}
\label{section-perfect-information}

We start with the core definitions of a game in extensive form with imperfect information, as typically encountered in literature \citep{Kuhn1950} \citep{Kuhn1953}.

What  differentiates a game with imperfect information from a game with perfect information is that, when making a choice, an agent may not be fully informed about the choices that other agents have already made. One interpretation thereof, which is the focus of our paper, is that these other decisions are being made in separate rooms, with no communication: Bob might already have decided, but Alice, who is on the other side of the wall, does not know what his decision was. Nodes are thus grouped into information sets, and a choice of action has to be taken by an agent, who does not know at which node, within this information set, they are playing.

As discussed in the previous section, games with imperfect information subsume games in normal form: a game in normal form can be expressed as a game in extensive form with imperfect information. In literature, these games can also be referred to as games with perfect information and simultaneous moves \citep{Rubinstein1994} \citep{Dubey1984}. Games with imperfect information also subsume games in extensive form with perfect information, which have only singleton information sets: this is the special case where the agents are perfectly informed about the decisions previously made by other agents, and thus they know exactly at which node they are playing.

\subsubsection{Formal definition}
\label{section-definition-game}

We take as the definition of a game in extensive form with imperfect information the same as that of \citet{Jackson2019} at Stanford, making a few more properties explicit. An alternate definition based on sequences of information sets and actions is given by \citet{Rubinstein1994} and is popular as well.

After this, we will give an example.

\begin{definition}[Game in extensive form with imperfect information]
A game $\Gamma$ in extensive form with imperfect information is a tuple $(N, A, H, Z, \chi, \rho, \sigma, u, I)$ where
\begin{itemize}
\item $N$ is a set of players.
\item $A$ is a set of actions (common across all players).
\item $H$ is a set of choice (non-terminal) nodes.
\item $Z$ is a set of outcomes.
\item $\chi \in H \rightarrow \mathcal{P}(A)$ is the action function, assigning each choice node to the set of available actions at that node.
\item $\rho \in H \rightarrow N$ is the player function, assigning each choice node to a player.
\item $\sigma \in H \times A \rightarrow H \cup Z $ is the successor function, assigning each pair of choice node and action (available at that choice node) to a choice node or outcome.

There are a few constraints on $\sigma$ to enforce that the game is a tree. First, it is injective. Second, $\sigma(h, a)$ is only defined if $a\in \chi(h)$. Lastly, $\sigma$ must organize the choice nodes and outcomes in a single connected component: there can only be one root, as opposed to a forest\footnote{This was implicit in the original definition.}.

\item $u \in N \times Z \rightarrow \mathbb{R}$ is the utility function, assigning each player and outcome to a payoff. If we are only interested in pure strategies, payoffs are ordinal, not cardinal, meaning that it only matters how they compare, but their absolute values do not. Literature thus also models utilities with an order relation, with no explicit payoff numbers. Using numbers improves readability and makes it easier to talk about examples.

\item $I$, the information partition, is an equivalence relation on $H$ that is compatible with the player function as well as with the action function. By convention the equivalence relation is expressed in terms of information sets, which are partitions $(I_{i,j})_{i,j\in N\times \mathbb{N}}$, one for each player and integer index. Formally, it fulfils for any $i$ and $j$ that $$\forall h \in I_{i,j}, \rho(h) = i$$ as well as $$\forall h, h' \in I_{i,j}, \chi(h) = \chi(h')$$
\end{itemize}
\end{definition}

\subsubsection{Example}

Figure \ref{figure-imperfect-information} shows an example of a game in extensive form and with imperfect information. Here, there are two players: Ulysses and Valentina: $N=\{U,V\}$.

There are 16 actions: $A=
\{a, b, c, d, e, f, g, h, i, j, k, l, m, n, o, p\}$.

There are 12 nodes: $N=\{n_1, ..., n_{12}\}$ and 13 outcomes: $Z=\{o_1, ..., o_{13}\}$.

Each node is associated with possible actions like so: $\chi(n_1)=\{ a, b \}$, $\chi(n_2)=\{ c, d \}$, $\chi(n_3)=\{ e, f \}$, etc..

Each node is associated with a player: $\rho(n_1)=U$, $\rho(n_2)=V$, $\rho(n_3)=V$, etc.

The nodes and outcomes are connected as a tree: $\sigma(n_1, a)=n_2$, $\sigma(n_1, b)=n_3$, etc.

The payoffs are specified like so: $u(U, o_1)=2$, $u(V, o_1)=4$, etc.

Information sets are specified as $I_{U,1}=\{n_1\}$, $I_{U,2}=\{n_4, n_5\}$, etc.

\begin{figure}
\centering\includegraphics[width=\textwidth]{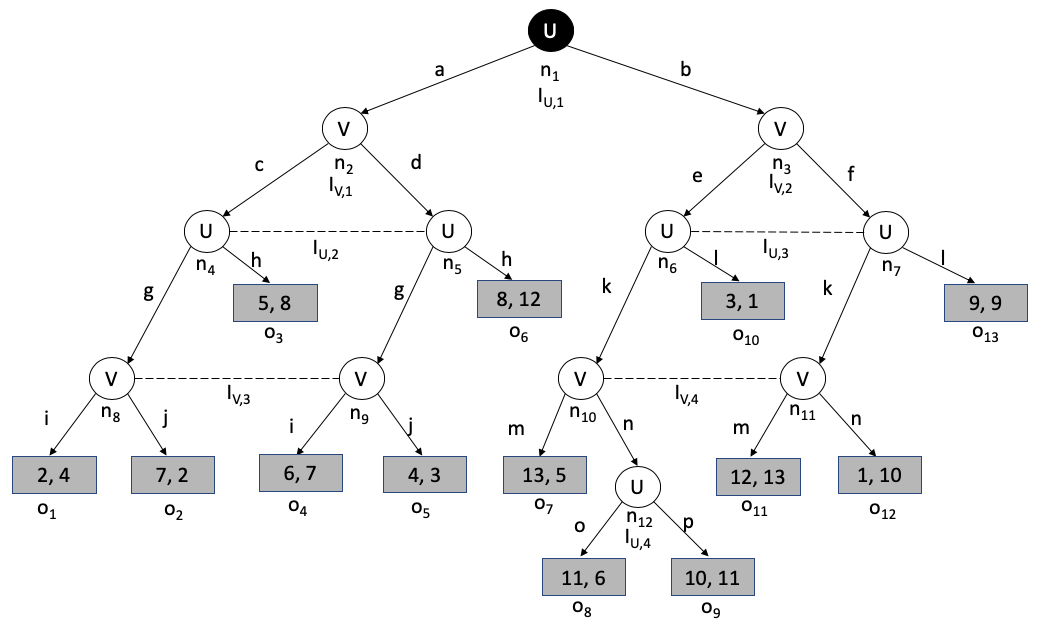}
\caption{A game in extensive form and with imperfect information. Information sets are denoted with dotted lines.}
\label{figure-imperfect-information}
\end{figure}

\subsection{A primer on Minkowski spacetime}
\label{section-spacetime}

\subsubsection{(High-level) mathematical definition}

We start with a mathematical definition, but also point out that it is not crucial for the reader to understand the math in detail. The crucial part is to understand the dependency graph that can exist between the locations of the decisions, which is a rather simple graph (see Section \ref{section-separation}).

Spacetime is modeled by physicists as a vector space $\mathbb{R}^n$, in practice, $\mathbb{R}^4$, on our example figures $\mathbb{R}^2$, where each point contains information about its ``distortion'', called a metric tensor. This is called a pseudo-Riemannian manifold.

The metric tensor is a non-degenerate, symmetric bilinear form (an $n\times n$ matrix). A bilinear form is a function that takes two points (vectors) and outputs a real number that can be interpreted as some sort of ``angle'' between them. Bilinear means that, if a point is fixed, the corresponding partial function varies linearly in the other point. Symmetric means that these metric does not depend on the order of the two points. Non-degenerate means that the determinant of its matrix is not zero. In classical physics, such bilinear forms are used to define an inner product in (three-dimensional) space, and they have to be positive definite: in Euclidian space, the bilinear form applied on the vector connecting two specific events, as both its left and right parameter (this is the associated squared norm) is always positive, and its square root is taken as the definition of Euclidian distance between these two events; this is classical geometry as most know it. However, as we will see shortly, in relativity theory, spacetime intervals (the ``Minkowski norm squared'') can be positive as well as negative.

In this paper, we furthermore only consider flat, pseudo-Euclidian manifolds, i.e., those in which the metric tensor is the same at all positions across space and time. In other words, there are no massive objects such as a black hole or a star that locally distort the structure of spacetime and we ignore gravity or consider it negligible. This is known as Minkowski spacetime \citep{Minkowski1908} and special relativity. 

In Minkowski spacetime, the spacetime interval between two points is independent of the observer, assuming that inertial timeframes (i.e., the perspectives that an observer can have) are obtained from each other by Lorentz transformations\footnote{They correspond to observers moving at constant speeds from one another. Acceleration and deceleration, equivalent to gravitation, are out of the scope of special relativity and require general relativity.}. We assume a metric signature $(n-1,1)$\footnote{``Lorentzian'' corresponds to having a 1 in the signature, i.e., only one dimension of time.}, having in mind that, in practice, this is commonly $(3,1)$\footnote{$(1,1)$ on our simplified figures.}. The first three coordinates are known as space, the last one as time. On our diagrams, though, we will only represent one dimension of space and one of time for ease of understanding.

Figure \ref{figure-primer} shows a representation of a portion of Minkowski spacetime with one dimension of space and one of time. It also shows four events put at various locations; each event is a decision. Given any pair of events, we can compute spacetime distances using the squared norm corresponding to our bilinear form, but we will only be interested, for our discussion, in the sign of that norm, and we will see that it is very easy to understand visually: in this figure, the spacetime distance between A and B is negative (and B is in the white area), and the spacetime distance between A and D on the one hand, but also between A and C on the other hand, is positive (and C and D are in a colored area). The rest of the figure (light cones, timelike/spacelike separation) can be ignored for now and will be discussed shortly.

\begin{figure}
\centering\includegraphics[width=\textwidth]{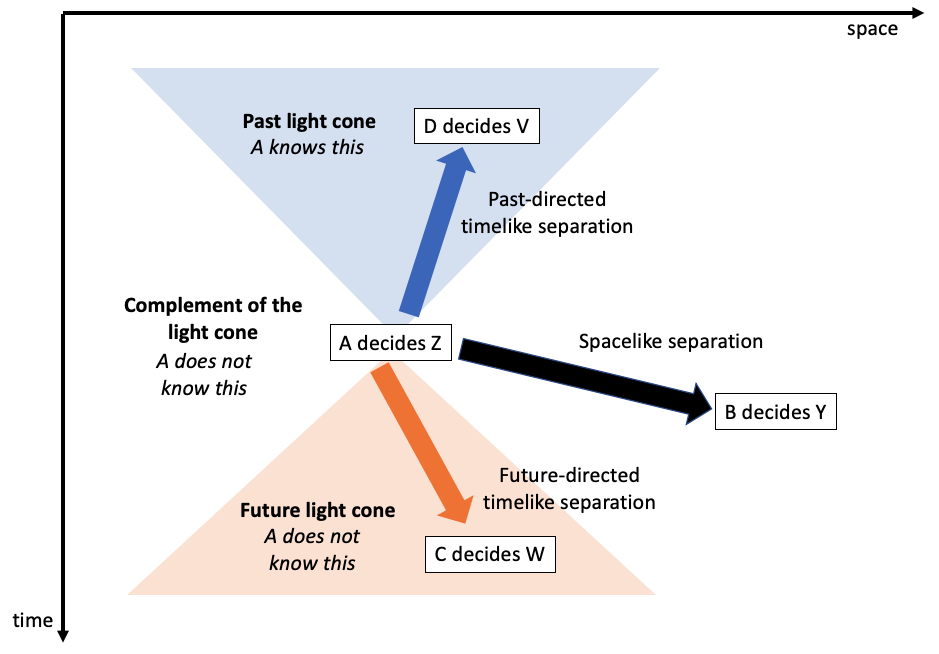}
\caption{A more abstract diagram illustrating timelike separation, spacelike separation, and the light cones. Four events are shown: A, B, C, and D, which are agents making decisions. The light cones around A are indicated in colors (but we could draw them for any other event as well) and we can see that D precedes A, A precedes C, whereas A and B occur in no specific order that everybody would agree with. From the perspective of decision theory and knowledge: A knows D's decision, C knows A's (and D's) decision, A does not know B's decision and neither does B know A's decision.}
\label{figure-primer}
\end{figure}

\subsubsection{Order of events, and timelike vs. spacelike separation}
\label{section-separation}

Given two events (this is the name given by physicists to vectors or points) in spacetime, we can calculate their spacetime interval with the bilinear form (i.e., its norm squared, applied to the difference between the events). The sign of this interval\footnote{Note that whether to assign + to timelike and - to spacelike or the opposite is a pure matter of convention. In this paper, we will also ignore the choice of sign convention and only refer to timelike and spacelike separation.} allows us to classify pairs of (distinct) events into one of two cases:

\begin{itemize}
\item either they are \emph{timelike}-separated, meaning that any observer (inertial timeframe) would see these two events occur after one another, and all observers agree on the order in which they occur. We can thus say that one of the two events \emph{precedes} the other because its time coordinates are smaller than the other event's time coordinates in any inertial timeframe. What is important is that \emph{everybody} agrees on this order of events, which is absolute.
\item or they are \emph{spacelike}-separated, meaning that the order in which these events occur depends on the observer. No signal can be sent between these two events because it would involve faster-than-light travel, equivalent to traveling back in time for some other observer. Spacelike separation, by analogy, can be thought of as being in different rooms with no way to communicate, except that the laws of physics, at least in our current understanding, actually \emph{enforce} that no communication can happen between two spacelike-separated events.
\end{itemize}

Figure \ref{figure-primer} illustrates this with four events, which are decisions made by the four agents Alice (A), Bob (B), Caroline (C) and Daniel (D). Alice is spacelike-separated from Bob, and is not informed of his decision: it would have to travel faster than light to reach her. Alice is timelike-separated from both Daniel's and Caroline's decision. More precisely, Daniel's decision precedes Alice's decision so that Alice is informed that Daniel picked V\footnote{Strictly speaking, it could be that Alice is not actually informed even though she could have known. But in our spacetime games, we assume perfect information, in the same sense as this term is commonly used in game theory so that everybody knows everything that is in their past light cone.} (and Daniel does not know what Alice will decide). Caroline's decision follows Alice's decision so that Alice does not know what Caroline will decide (but Caroline knows Alice's decision.

From that point on, it is enough to remember that two distinct events in spacetime are either spacelike or timelike separated. Understanding the details of the metric tensor is not necessary to understand our contribution. An event is simply a point in spacetime at which a decision is made.

There are only three possible cases given two decisions, made by, say, Alice and Bob, at two points in spacetime: (i) Alice and Bob don't know each other's decisions (spacelike separation); (ii) Bob knows what decision Alice made (timelike separation and Alice's decision precedes Bob's decision); (iii) Alice knows what decision Bob made (timelike separation and Bob's decision precedes Alice's decision). 

\subsubsection{Past and future light cones}

Given an event $E$, all points in spacetime that are timelike-separated from $E$ and precede $E$ (and thus, all observers agree on this) form the so-called past light cone of $E$. All points in spacetime that are timelike-separated from $E$ and follow $E$ (and thus, all observers agree on this) form the so-called future light cone of $E$. The union of the two light cones is referred to as the light cone of $E$. The points that are spacelike-separated from $E$ are referred to as the complement of the light cone and, for these, not all observers will agree on whether they precede or follow $E$. This is also shown in Figure \ref{figure-primer}.

Any agent, say, Alice, placed at $E$ can only know at most what happened in $E$'s past light cone because anything else cannot possibly reach Alice with a speed slower than that of light. Since all observers agree that all events in the past light cone of $E$ precede $E$, this is consistent with Alice's knowing what happened in any point in the past light cone of $E$, as it always precedes $E$. Furthermore, in this paper, we assume that Alice at $E$ is informed of \emph{everything} in its past light cone: in other words, perfect information (but with an understanding thereof adapted to spacetime).

\section{Spacetime games with perfect information}

Now that we have introduced the basics of (flat) spacetime, we can explain how games can be set up on this playground. We will start with two simple cases, which actually correspond to games in normal form and games in extensive form with perfect information, and then generalize to an arbitrary configuration that subsumes these two cases.

\subsection{Special cases}

\subsubsection{Games in normal form and spacelike separation}

Games in normal form are classically described as decisions made in different rooms (or simultaneously). This naturally mapped to what we have previously described as spacelike-separation. Indeed, as we saw, no transfer of information is possible between two spacelike-separated events.

We have already shown how such a game, like the prisoner's dilemma, can be played with one agent on the Earth and one on Mars (Figure \ref{figure-spacelike}), but now we can describe this situation with more precise vocabulary. Indeed, spacelike-separation is enforced if the decisions are taken within less than four minutes of each other: no signal can travel between Earth and Mars faster than this.

Figure \ref{figure-spacelike-diagramm} shows a more abstract depiction of this game where we show the two decisions in a simplified spacetime with only one dimension of space. We show the future light cone of each of the agents' making a decision (lines going to the bottom-left and bottom-right of each event). People within the future light cone of a decision are informed of this decision. But since Alice's decision and Bob's decision are not within each other's future light cones, they do not know about each other's decisions.

The notation $I_{A,1}$ that we introduce in this figure means that this is Alice (A)'s first (1) decision. Likewise, $I_{B,1}$ is Bob's first decision. $I_{A,2}$ would be used for a hypothetical second decision made (later) by Alice, and so on.

\begin{figure}
\centering\includegraphics[width=0.7\textwidth]{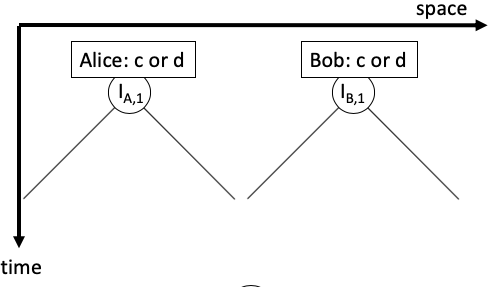}
\caption{A more abstract diagram showing the positions of Alice and Bob in spacetime, but showing only one dimension of space for the sake of simplicity.}
\label{figure-spacelike-diagramm}
\end{figure}

Finally, Figure \ref{figure-spacelike-dependence} shows another representation that is called a causal dependency graph. On a dependency, each node is an event, and two events are attached with an edge if they are timelike-separated, which means that the second node is informed of what happened at the first node. By convention and for ease of display, we only show the minimalistic version of our causal dependency graphs, called its transitive reduction (it is obvious, even though mathematically not trivial, that if $A$ precedes $B$ and $B$ precedes $C$, then $A$ precedes $C$). The graph for the game at hand, however, has no edges anywhere because there is no timelike separation. Alice and Bob are not informed of each other's decisions.

Also, we need to add for the sake of completeness that, in this configuration, Alice makes her decision no matter what, and Bob makes his decision no matter what. This is not always true.

\begin{figure}
\centering\includegraphics[width=0.3\textwidth]{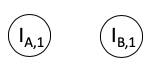}
\caption{The dependency graph of Alice's and Bob's decisions in the game in normal form. There are no edges because there is no dependency (spacelike separation).}
\label{figure-spacelike-dependence}
\end{figure}

\subsubsection{Games in extensive form with perfect information, and timelike separation}

Games in extensive form and with perfect information are classically described as decisions made in turn, which implies an elapse of time and flow of information. This naturally mapped to what we have previously described as timelike-separation. Indeed, as we saw, there is a transfer of information from each move to the subsequent moves. Also, a future node in the game will only be reached if specific moves have been made in its past.

We have already shown how such a game, like the promise game, can be played with two agents on Earth  playing in turn (Figure \ref{figure-timelike}), but now we can describe this situation with more precise vocabulary. Indeed, timelike-separation is enforced if each decision is taken by an agent after the other relevant decisions have been made, and once this agent was informed about them.

Figure \ref{figure-timelike-diagramm} shows a more abstract depiction of this game where we show the two decisions in a simplified spacetime with only one dimension of space. We show the future light cone of each of the agent's making a decision (lines going to the bottom-left and bottom-right of each event). People within the future light cone of a decision are informed of this decision. In this case, Alice's decision ($I_{A,1}$) precedes Bob's decision ($I_{B,1}$), i.e., Bob's decision is in Alice's decision's future light cone.

But there is something more on this figure that we have not seen in the discussion of games in normal form: Bob only gets to make his decision if Alice's decision was $c$. This is very natural in games in extensive form: for example, a Chess player can only move their Queen if the Queen was not previously captured by a move made by the opponent. In the extensive form, a decision at a node is only effective if this node is reached, that is, if the players in the past played \emph{towards} that node. As we will see later, this additional rectangle (``If Alice: c'') is called the contingency coordinates of Bob's decision, i.e., the pre-conditions for this decision to be made\footnote{This does not mean that Bob cannot have a plan, or specific intent, for this decision should it be reached, or should it have been reached. This is commonly done in Nash game theory and known as strategies.}.

\begin{figure}
\centering\includegraphics[width=0.5\textwidth]{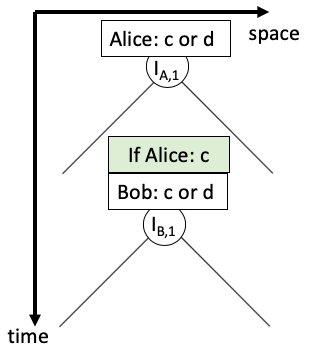}
\caption{A more abstract diagram showing the positions of Alice and Bob in spacetime, but showing only one dimension of space for the sake of simplicity.}
\label{figure-timelike-diagramm}
\end{figure}

Finally, Figure \ref{figure-timelike-dependence} shows the associated causal dependency graph. We can see here that there is an edge going from Alice's decision to Bob's decision (we did not represent the contingency coordinates on our causal dependency graphs, but they nevertheless exist). Actually, the causal dependency graph of a game in extensive form with perfect information is simply (and naturally) isomorphic to its tree representation, but without the outcomes, i.e., leaves. Likewise, a representation in spacetime can easily be obtained by arranging the nodes in the order of a depth-first (or breadth-first) traversal of the tree along the time axis. The contingency coordinates of any node are also naturally taken from all decisions that must be made in ancestor nodes (in the game's tree representation) for this node to be reached.

\begin{figure}
\centering\includegraphics[width=0.1\textwidth]{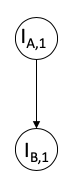}
\caption{The dependency graph of Alice's and Bob's decisions in the game in extensive form. There is an edge from Alice's decision to Bob's decision because Bob is informed of Alice's decision, and only actually makes his decision if Alice decided in a specific way (timelike separation).}
\label{figure-timelike-dependence}
\end{figure}

The case should now be clear for games in normal form, or in extensive form with perfect information. How about more generic positions in spacetime? We are now going to generalize our insights and build the spacetime game paradigm.

\subsection{Decision points}

\subsubsection{Decision making}

Assume, as before, that we have a Minkowski spacetime in place. We can place in it, at various locations, events that we call decision points. At each decision point, an agent, taken from a set $N$ of agents\footnote{As will quickly become apparent, we do use the same letters as in the definition of games in Section \ref{section-perfect-information} because of the (intended) correspondence.}, potentially\footnote{Potentially, because some preconditions on past decisions may have to be met -- we called these contingency coordinates.} makes a choice amongst a (specified) subset of a global set of actions $A$.

In the previous subsections, we looked at specific cases in which the locations of the decision points were all spacelike separated (equivalent to a game in normal form), or in which they were all timelike separated (equivalent to a game in extensive form with perfect information). Now, we relax this condition and allow for the decision points to be placed \emph{anywhere} in spacetime\footnote{There may be some obvious restrictions due to the impossibility of two persons to actually stand at the same position in space and at the same time, which we will discuss later on.}.

\begin{definition}[Decision point]
Given a set of agents $N$, a set of actions $A$, and a Minkowski spacetime manifold (usually $\mathbb{R}^4$ endowed with a (3,1) metric), a decision point $\hat{I}_{i,j}$ consists of (i) a specific agent $i\in N$ making a decision, of (ii) a (sub)set of actions $\hat{\chi}(\hat{I}_{i,j})\subseteq A$ that agent $i$ can pick from and of (iii) a specific location in spacetime, which we denote $\mathcal{L}(\hat{I}_{i,j})\in \mathbb{R}^4$, at which this decision is made.

\begin{figure}
\centering\includegraphics[width=0.7\textwidth]{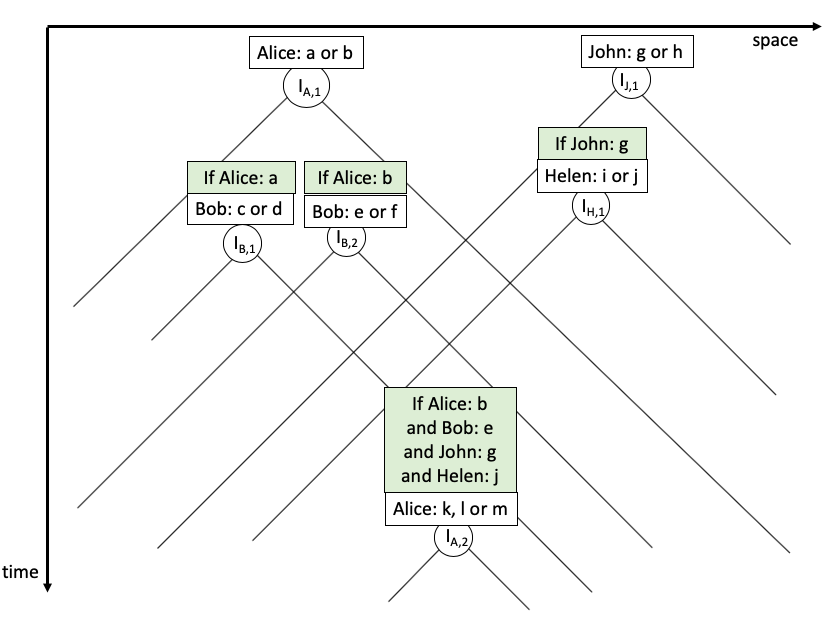}
\caption{Six decisions points, located in Minkowsky spacetime. Lines indicate their future light cones to make timelike and spacelike separation visible.}
\label{figure-spacetime-example}
\end{figure}

Decision points are indexed with two subscripts: the agent $i$ making the decision, and an integer $j\in\mathbb{N}^*$ that indexes all the decision points at which agent $i$ makes a decision\footnote{By convention increasing integers, but this second subscript is completely arbitrary and does not need to be in any relationship with the positions in spacetime.}.
\end{definition}

In the following, we call $\hat{I}$ the set of all decision points for the configuration we are considering. Figure \ref{figure-spacetime-example} shows an example with six decision points located in spacetime (here $\mathbb{R}^2$), four agents (Alice (A), Bob (B), John (J), Helen (H)) and thirteen actions (a, b, c, d, e, f, g, h, i, j, k, l, m). The spacetime coordinates correspond to their positions in the figure, in which we also made future light cones visible. Figure \ref{figure-spacetime-example} also additionally contains contingency coordinates, but we will introduce their formalism in the next section.

\subsection{Timelike partial order and causal dependency graph}

Since each decision point has a location in spacetime, given two decision points, we can thus say that they are either spacelike-separated or timelike-separated based on their coordinates\footnote{If the distance is zero and the positions are different, it means that one can be accessed from the other only by traveling at light speed. We thus consider them to be timelike-separated because information can be sent from one to the other at the velocity of light. If, however, the two points are at the same location, we consider them to be spacelike separated -- because agents are made of fermions, i.e., put simply, particles that cannot co-exist at the same location, this would only physically make sense if these coinciding decision points actually take place under different conditions on the past, for example, one if an agent picked C earlier, and the other if that agent picked D earlier so that this is a matter of convention.}. Timelike separation induces a partial order on decision points, which models that a decision point occurs after another one for any observer.  We formally define it now.

\begin{definition}[Timelike precedence ($\prec$)]
Given a set of agents, a set of actions, a Minkowski spacetime manifold and a set of decision points located in it, if two decision points are timelike-separated, then one of them occurs before the other one. We use the partial order $\prec$ to denote timelike-separation according to this order, i.e., $\hat{I}_{i,j}\prec\hat{I}_{k,l}$ if $\hat{I}_{i,j}$ and $\hat{I}_{k,l}$ are timelike-separated, and the former occurs before the latter\footnote{There is the special case of comparing a decision point with itself, and we could say by convention that it both precedes and follows itself for $\prec$ to be mathematically and formally reflexive. However, we do not need in this paper to compare a decision point with itself.}.
\end{definition}

\begin{figure}
\centering\includegraphics[width=0.5\textwidth]{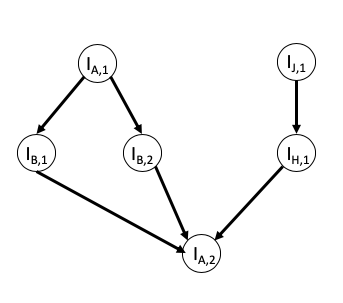}
\caption{The causal dependency graph (which is directed acyclic graph) representing the timelike separation partial order $\prec$ on our example.}
\label{figure-spacetime-example-timelike-dependence}
\end{figure}

The partial order is completely and unambiguously defined by the spacetime locations of the decision points. In practice, the exact locations are irrelevant and only this partial order matters. The partial order can visually be represented as a Directed Acyclic Graph (DAG) called a causal dependency graph. Figure \ref{figure-spacetime-example-timelike-dependence} shows, for the same example, the causal dependency graph that makes the partial order on timelike separation explicit, based on the light cones drawn in Figure \ref{figure-spacetime-example}.

\subsection{Contingency coordinates}

When a game is played, it depends on the past whether a specific decision is made or not. This is familiar to anybody who uses games in extensive form: a player only actually plays at a specific node $n$ if, in the past, all previous players made decisions that led precisely to that node $n$. The same applies to spacelike games: a decision has preconditions, which we call contingency coordinates.

\subsubsection{Preconditions for a decision}

Let us go back to the example of Figure \ref{figure-spacetime-example}. In this spacetime game, Bob will only actually make a decision and pick between c and d if Alice formerly picked a. Likewise, Bob will only actually make a decision and pick between e and f if Alice formerly picked b. Helen will only actually make a decision and pick between i and j if John formerly picked g. Finally, Alice will only actually make a decision and pick between k, l and m if she previously picked b \emph{and} Bob picked e \emph{and} John picked g \emph{and} Helen picked j.

We call the set of decisions that must have been made in the past for another decision $\hat{I}_{i,j}$ to take place the \emph{contingency coordinates} of $\hat{I}_{i,j}$. 

So, to fully describe the decisional configuration that we want to model, we must provide, for each decision point, in addition to its spacetime coordinates, its contingency coordinates. Formally, the contingency coordinates of a decision point are the actions that must be taken at previous, timelike-separated decision points for this decision point to actually be reached. In other words, the ability to make that decision must be caused\footnote{The notation of causality, in this paper, coincides with timelike-separation.} by a specific chain of events at decision points preceding it in spacetime.

First, we need a formal definition of a raw assignment of decisions, which will also be useful for further definitions such as histories and strategies.

\begin{definition}[Raw assignment of decisions]
Given a set of agents, a set of actions, a Minkowski spacetime manifold, a set of decision points located in it and their associated possible actions denoted with the mapping $\hat\chi$, a raw assignment of decisions is a partial function\footnote{$A\supset \mspace{-3mu} \rightarrow B$ denotes the set of partial functions from A to B, i.e., not every element in A is associated with an element in B.}  $\alpha\in\hat{I} \supset \mspace{-3mu} \rightarrow A$ that associates some decision points, e.g., $\hat{I}_{i,j}$ with an action $\alpha(\hat{I}_{i,j})\in\hat\chi(\hat{I}_{i,j})$.
\end{definition}

A raw assignment of decisions is thus a partial function from the set of decision points to the set of actions. If no action is assigned to a specific decision point, we denote it $\alpha(\hat{I}_{k,l})=\perp$. As a consequence, giving two raw assignments, we can use function terminology and say that one can be a restriction or an extension of the other, meaning that they coincide on the intersection of their domains, and the domains are in a subset/superset relationship. Given two raw assignments that perfectly match on their domains, we can also define their union by noticing that a partial function is a relation, and that a relation is a subset of $\hat{I}\times A$. The union of the two correspond sets is then also a partial function from $\hat{I}$ to $A$ because the original partial functions match on the intersection of their supports.
 
Figure \ref{figure-spacetime-example-assignments} shows, at the top, some arbitrary raw assignment that does not actually mean anything, for the purpose of illustration.

\begin{figure}
\centering\includegraphics[width=\textwidth]{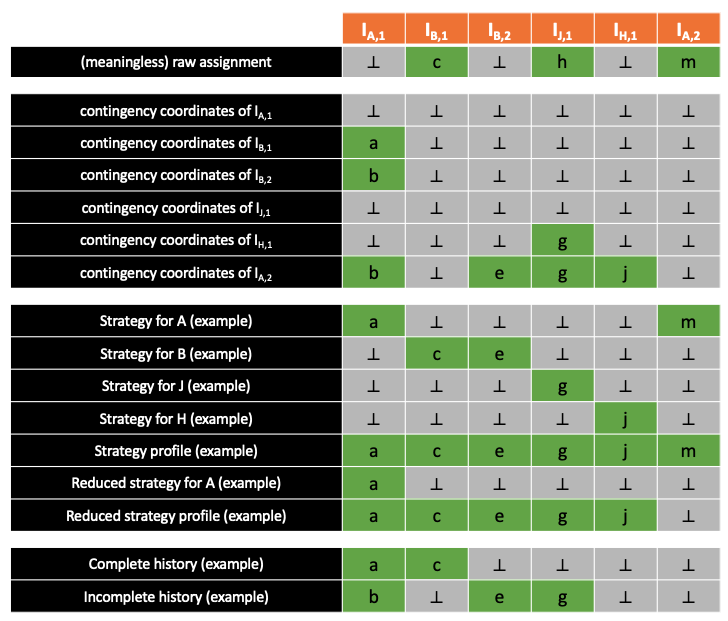}
\caption{Many examples of various assignments on our running example. We show first a meaningless random assignment. Then, the contingency coordinates of our six decision points, consistent with Figure \ref{figure-spacetime-example}. We also show strategies, a strategy profile, a reduced strategy, a reduced strategy profile, a complete history, and an incomplete history, which are all, formally, assignments. Note that the order in which we show the six decision points was arbitrarily chosen, but to avoid any confusion, we picked an order that respects timelike precedence (this is called a linearization of the timelike precedence DAG, as we will see later).}
\label{figure-spacetime-example-assignments}
\end{figure}

Now, we can define the contingency coordinates of a decision point as a particular raw assignment. Many more particular classes of raw assignments, with their own different restrictions, will be introduced later.

\begin{definition}[Contingency coordinates]
Given a set of agents, a set of actions, a Minkowski spacetime manifold, a set of decision points located in it (with $\hat\chi$ denoting the mapping of decision points to choosable actions), and a specific decision point $\hat{I}_{i,j}$ in that set, the contingency coordinates of decision point $\hat{I}_{i,j}$ are a raw assignment of decisions denoted denoted $\gamma_{i,j}$ that is only defined on decision points that timelike-precede $\hat{I}_{i,j}$.

Semantically, agent $i$ only gets to make their $j$-th decision if agent $k$ previously picked $\gamma_{i,j}(\hat{I}_{k,l})$ as their $l$-th decision, and this must be true for every $k$ and $l$ for which the assignment is defined.
\end{definition}

On our example (Figure \ref{figure-spacetime-example}), formally, $\gamma_{A,1}(\hat{I}_{k, l})=\perp$ for any decision point $\hat{I}_{k, l}$ because Alice makes her first decision in any case. $\gamma_{B,1}(\hat{I}_{A, 1})=a$ (and it is $\perp$ for all other decision points) because Bob only makes his decision if Alice previously picked a. And for Alice's second decision, our requirements translate to $\gamma_{A,2}(\hat{I}_{A, 1})=b$, $\gamma_{A,2}(\hat{I}_{B, 1})=e$, $\gamma_{A,2}(\hat{I}_{J, 1})=g$, and $\gamma_{A,2}(\hat{I}_{H, 1})=j$. In general, the family of functions $(\gamma_{i,j})_{i,j}$ can directly be read from a diagram such as Figure \ref{figure-spacetime-example}. They are shown in Figure \ref{figure-spacetime-example-assignments}, below the arbitrary raw assignment.

With contingency coordinates, we can express a more fine-grained precedence relation: a decision point might precede another decision point, but if their contingency coordinates mismatch, there is no actual realization of a world in which both decisions are made. For example, there is no world in which Bob picks $d$ at $\hat{I}_{B, 1}$ and Alice picks $k$ at $\hat{I}_{A, 2}$ because Alice's picking $k$ at $\hat{I}_{A, 2}$ requires her having previously picked $b$, whereas Bob's picking d at $\hat{I}_{B, 1}$ requires Alice having previously picked $a$ before. Therefore $\hat{I}_{B, 1}$ does not \emph{actually precede} $\hat{I}_{A, 2}$.

\begin{definition}[Actual precedence]
Given a set of agents, a set of actions, a Minkowski spacetime manifold, a set of decision points located in it (with $\hat\chi$ denoting the mapping of decision points to choosable actions), and two specific decision points $\hat{I}_{i,j}$ and $\hat{I}_{k,l}$, we say that $\hat{I}_{k,l}$ actually precedes $\hat{I}_{i,j}$ if (i) $\hat{I}_{k,l}\prec\hat{I}_{i,j}$ (i.e., the former timelike-precedes the latter) and (ii) $\gamma_{k,l}$ is a restriction of $\gamma_{i,j}$, which means formally:

$$\forall \hat{I}_{m, n}, \gamma_{k,l}(\hat{I}_{m, n})\neq\perp \implies \gamma_{k,l}(\hat{I}_{m, n}) = \gamma_{i,j}(\hat{I}_{m, n}) $$
\end{definition}

Actual precedence is thus a stricter relation than timelike precedence and corresponds to the fact that there is a (realizable) world with a chain of events in which both decisions happen one after the other. Figure \ref{figure-spacetime-example-actual-dependence} shows the actual dependence graph for our running example.

\begin{figure}
\centering\includegraphics[width=0.5\textwidth]{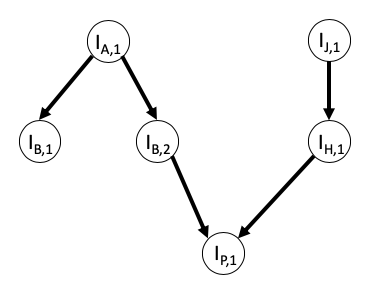}
\caption{The dependency graph corresponding to actual precedence according to the contingency coordinates. It is a subset of the causal dependency graph of Figure \ref{figure-spacetime-example-timelike-dependence}.}
\label{figure-spacetime-example-actual-dependence}
\end{figure}

There is one last step: some assignments of contingency coordinates make no sense. For example, we saw that $\hat{I}_{B, 1}$ does not actually precede $\hat{I}_{A, 2}$. But what if $\gamma_{A,2}(\hat{I}_{B, 1})\neq\perp$? It makes no sense to say that Alice gets to make her second decision if Bob's decision at $\hat{I}_{B, 1}$ was this or that because both can never happen jointly. There is thus a restriction on the contingency coordinates of any decision point: they can only be, and must be, defined on the decision points that actually precede it.

\begin{definition}[Consistent contingency coordinates]
Given a set of agents, a set of actions, a Minkowski spacetime manifold, a set of decision points located in it, carrying contingency coordinates $(\gamma_{i,j})_{i,j}$, we say that the contingency coordinates $(\gamma_{i,j})_{i,j}$ are consistent if, for any two decision points $\hat{I}_{i,j}$ and $\hat{I}_{k,l}$, $\gamma_{i,j}(\hat{I}_{k,l})\neq\perp$ if and only if $\hat{I}_{k,l}$ actually precedes $\hat{I}_{i,j}$
\end{definition}

We require that, for any spacetime game with perfect information, all contingency coordinates are consistent. If a game has inconsistent contingency coordinates, then it can be made consistent by pruning decision points that can never be reached, i.e., that have inconsistent contingency coordinates.

The contingency coordinates shown in Figure \ref{figure-spacetime-example-assignments} fulfill this consistency criterion.

\subsection{Histories}

If we now start asking the various involved agents to make their decisions according to the configuration we have just put in place, we will get an assignment of decisions to (actually reached) decision points. As described in the previous section, a decision is only made if the preconditions for it are fulfilled, which means that not necessarily \emph{all} decision points will be assigned a chosen action.

Such a list of decisions is one possibility of how things can happen: one possible world. We now look at \emph{all} possible worlds that can be instantiated from the decision points, i.e., considering all possible actions that can be chosen consistently.

A possible history is defined as a raw assignment of decisions that is compatible with their contingency coordinates.

\begin{definition}[History]
Given a set of agents, a set of actions, a Minkowski spacetime manifold, a set of decision points located in it, carrying contingency coordinates $(\gamma_{i,j})_{i,j}$, a history is a raw assignment of decisions $h$ that is consistent with the contingency coordinates of each decision point on which it is defined, in the sense that $h$ is an extension of the contingency coordinates of all decision points on which it is defined. Formally:

$$\forall \hat{I}_{i,j}, \hat{I}_{k,l}, h(\hat{I}_{i,j})\neq \perp \wedge \gamma_{i,j}(\hat{I}_{k,l})\neq \perp \implies h(\hat{I}_{k,l})=\gamma_{i,j}(\hat{I}_{k,l})$$.

A history is said to be \emph{complete} if it cannot be extended to a more complete history, i.e., that assigns decisions to strictly more decision points. Otherwise, it is said to be \emph{incomplete}.
\end{definition}

Figure \ref{figure-spacetime-example-assignments} shows both an example of a complete history and an incomplete history.

Note that, in any history, complete or incomplete, for any two decision points that are assigned a decision, if one timelike-precedes the other, then it follows that the former also actually precedes the latter. This directly follows from the consistency criterion of the history. A history thus corresponds to an actual DAG of events happening in the same possible world.

\subsection{Payoffs}

A complete history corresponds to a possible complete play of the game. Since it is consistent, only decisions that are actually reached are made, and since it is complete, all decisions that are actually reached are made. For each complete history, we can associate payoffs that model the preference of the agents between different complete histories.

Formally, we model the agent's preferences with a function $u$ mapping each complete history and agent to a number. In a scenario in which we only consider ordinal payoff semantics, it is equivalent to a total order relation over complete histories for each agent.

\begin{definition}[Payoffs]
Given a set of agents, a set of actions, a Minkowski spacetime manifold, a set of decision points located in it, carrying contingency coordinates, a utility function is a function $u$ that assigns each decision point and agent to a real number: $u\in\hat{I}\times N\rightarrow\mathbb{R}$

\end{definition}

\subsection{Formal definition}

We can now summarize and formally define what a spacetime game is (which is mostly a wrap up of the concepts introduced hitherto).

\begin{definition}[Spacetime game with perfect information]
A spacetime game with perfect information is a tuple $\Gamma=(N, A, \hat{I}, \prec, \hat{\chi}, \gamma, u)$ comprising:
\begin{itemize}
\item a set of agents $N$;
\item a set of actions $A$;
\item a set of decision points $(\hat{I}_{i, j})_{i\in N,j\in \mathbb{N}}$, each located in spacetime, with their relative locations encapsulated in a timelike-precedence relation $\prec$;
\item a mapping from decision points to a set of actions $\hat\chi\in\hat{I}\rightarrow\mathcal{P}(A)$;
\item a mapping $(\gamma_{i,j})_{i,j}$ associating each decision point $\hat{I}_{i, j}$ to its contingency coordinates\footnote{i.e., a raw assignment that fulfils the contingency coordinate constraints.}, i.e., for any $i,j$, $\gamma_{i,j}\in\hat{I}\supset \mspace{-3mu} \rightarrow A$ with the constraint that the consistency condition on $\gamma$ for contingency coordinates is met;
\item and payoff distributions $u$ assigning a real number to each agent as well as complete history\footnote{i.e., a raw assignment in $A^{\hat{I}}$that fulfils the complete history constraints.}. 
\end{itemize} 
\end{definition}

\section{Strategic form of a spacetime game}

Any spacetime game with perfect information can be naturally associated with an equivalent game in normal form. This is similar to how games with imperfect information can also all be associated with an equivalent game in normal form (called its strategic form). The strategic form of a spacetime game can be seen as an interpretation of the game in which the players decide and plan, in advance and in separate rooms, of what they would do for any possible game play\footnote{which in turn can be seen as a spacetime game with only spacelike-separated decisions.}.

\subsection{Strategies}

The strategy of an agent can be defined as a choice of action for each one of their decisions (actual or not).

\begin{definition}[Strategy]
Given a spacetime game  $\Gamma=(N, A, \hat{I}, \prec, \hat{\chi}, \gamma, u)$, a strategy for agent $i\in N$ is a raw assignment that assigns an action to all decision points at which this agent decides $(\hat{I}_{i,j})_j$ and only those. 

The set of all such strategies for a given agent constitutes their strategy space $\Sigma_i$.
\end{definition}

Figure \ref{figure-spacetime-example-assignments} shows a few strategies, one for each agent.

\subsection{Strategy profile}

If all agents pick a strategy, this jointly forms a strategy profile.

\begin{definition}[Strategy profile]
A strategy profile is a raw assignment that maps every decision point to an action. The space of strategy profiles is thus a subset of $\hat{I}\rightarrow A$. If each agent picks a strategy from their strategy space, then a strategy profile can be built by combining all these strategies into a complete raw assignment, meaning formally that we take the union of all strategies if we see them as relations between the sets $\hat{I}$ and $A$.

The space of strategy profiles is thus the cartesian product of all agents' strategy spaces $(\Sigma_i)_i$.
\end{definition}

Any strategy profile uniquely determines a complete history, which in turn specifies the payoffs of a game in normal form, with the strategy spaces inferred as explained above. Figure \ref{figure-spacetime-example-assignments} shows the strategy profile corresponding to the four strategies shown in the rows above it, as well as the unique complete history that is a restriction of this strategy profile.

\begin{theorem}[Strategic form of a spacetime game with perfect information]
Given a spacetime game and a strategy profile, there is a unique complete history that is a restriction of this strategy profile. Thus, any strategy profile can be associated with the payoffs corresponding to this complete history.
This implicitly associates a strategic game (in normal form) with any spacetime game with perfect information. 
\end{theorem}

\begin{proof}[Unique complete history]
Let us take two complete histories $h$ and $h'$ that are restrictions of a given strategy profile $\Sigma \in \hat{I}\rightarrow A$ and show that they are the same.

First, it has to be the case that the two histories perfectly match on the intersection of their support, as they have to then both associate any decision point $\hat{I}_{i,j}$ in this support with $\Sigma(\hat{I}_{i,j})$.

Let us now define $h''$ as the union of these two complete histories\footnote{Remember that we previously defined what we mean by union of two raw assignments that perfectly match on the intersection of their supports.}. It directly follows that $h''$ is also a restriction of $\Sigma$.

Because $h$ and $h'$ are histories, we have:

$$\forall \hat{I}_{i,j}, \hat{I}_{k,l}, h(\hat{I}_{i,j})\neq \perp \wedge \gamma_{i,j}(\hat{I}_{k,l})\neq \perp \implies h(\hat{I}_{k,l})=\gamma_{i,j}(\hat{I}_{k,l})$$.

$$\forall \hat{I}_{i,j}, \hat{I}_{k,l}, h'(\hat{I}_{i,j})\neq \perp \wedge \gamma_{i,j}(\hat{I}_{k,l})\neq \perp \implies \hat{h'}(\hat{I}_{k,l})=\gamma_{i,j}(\hat{I}_{k,l})$$.

Now, if we take two decision points $\hat{I}_{i,j}, \hat{I}_{k,l}$ such that 

$$h''(\hat{I}_{i,j})\neq \perp \wedge \gamma_{i,j}(\hat{I}_{k,l})\neq \perp$$

Then it implies that either $h(\hat{I}_{i,j})\neq \perp$ or $h'(\hat{I}_{i,j})\neq \perp$ by definition of $h''$. Let us assuming it is $h$ without loss of generality. But then it follows from $h$ being a history that 

$$h''(\hat{I}_{k,l})=h(\hat{I}_{k,l})=\gamma_{i,j}(\hat{I}_{k,l})$$

This is true for any two decision points, and thus $h''$ is also a history. But since both $h$ and $h'$ are complete histories, they cannot be extended to more complete histories and it follows that $h=h''=h'$. Thus, there can only be at most one complete history that is a restriction of a strategy profile.
$\square$
\end{proof}

The complete history shown in \ref{figure-spacetime-example-assignments} is the (unique) complete history corresponding to the strategy shown above.

Figure \ref{figure-strategic-form} shows what our running example looks like if we build all possible strategies for each agent, obtain all possible strategy profiles. The figure shows some (arbitrarily) chosen payoffs -- note that the payoffs are consistent, in the same that two strategy profiles corresponding to the same complete history have the same payoffs.

\begin{figure}
\centering\includegraphics[width=\textwidth]{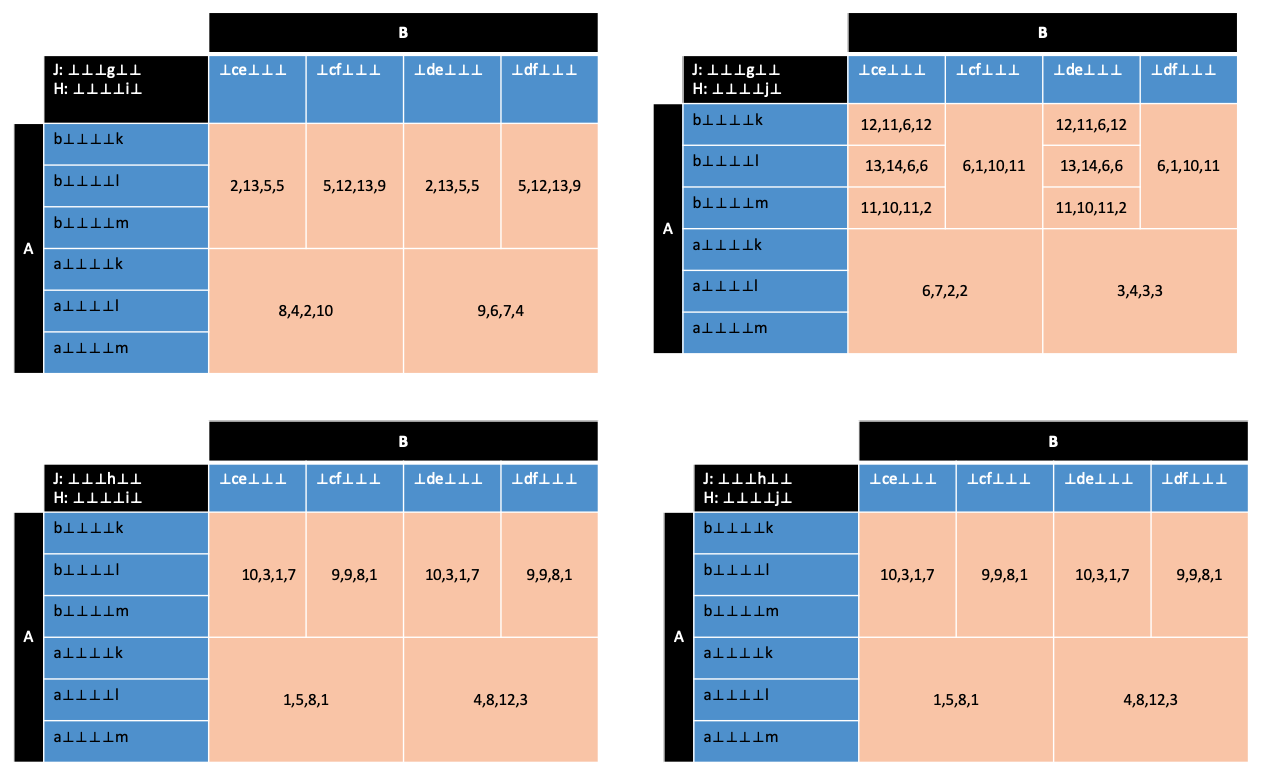}
\caption{The running example, shown in strategic form. For ease of view, we merged cells that correspond to the same complete history (and payoffs) whenever this was possible. There are 14 possible different complete histories and payoff distributions. As the matrix is four-dimensional (four agents), we show four separate planes for fixed strategies of John and Helen, where Alice is on the rows and Bob is on the columns. As can easily be seen, the last three strategies for Alice are equivalent. This will be addressed in the reduced strategic form.}
\label{figure-strategic-form}
\end{figure}

\subsection{Reduced strategic form}

The strategic form of a spacetime game can be simplified by eliminating redundant strategies\footnote{This is similar to the reduced strategic form of games in extensive form with imperfect information.}. Two  strategies are equivalent if they differ at a decision point that cannot be actually reached given the other choices of that same agents at other decision points that timelike-precede it.

\begin{definition}[Reduced strategy]
Given a strategy $\Sigma_i$ for agent $i$ and a decision point $\hat{I}_{i,j}$ in its support, we say that this decision point is not actual if its contingency coordinates conflict with $\Sigma_i$ on at least another decision point at which agent $i$ (formerly) decides.
Two strategies $\Sigma_i$ and $\Sigma'_i$ are equivalent if they only differ at not actual decision points. A canonical reduced strategy can be obtained for each suchly obtained equivalent class by un-assigning decision points that are not actual. A reduced strategy is thus a restriction of the original strategy.
\end{definition}

Figure \ref{figure-spacetime-example-assignments} shows the reduced strategy of Alice (A), where the second decision has been removed because it is irrelevant. It also shows the reduced strategic form.

\begin{theorem}[Unique reduced strategic form of a spacetime game]
Given a spacetime game with perfect information, if we select, for each agent, one of their reduced strategies and build the corresponding reduced strategy profile as their union, there exists a unique complete history that is a restriction of the  reduced strategy profile. Thus, a game in normal form can be built based on the reduced strategy spaces, and it is called the reduced strategic form of the spacetime game.
\end{theorem}

\begin{proof}
A reduced strategy of an agent is an equivalence class on the entire strategy space of that agent. Thus, we can pick a strategy of this agent in this equivalent class, and it is an extension of the reduced strategy. If we do it for all agents, we obtain a strategy profile $\sigma$, and we showed that there is a unique complete history $h$ that matches this strategy profile on its support ($h\subseteq\sigma$).

Now, if we take a decision point $\hat{I}_{i,j}$ in the support of the reduced strategy $\sigma_i$ of an agent $i$, it is by definition an actual decision point, and it implies that its contingency coordinates agree with all other decisions of the agent in this reduced strategy $\sigma$.

If the complete history $h$ assigns an action to $\hat{I}_{i,j}$,i.e.

$$h(\hat{I}_{i,j})=a$$

then we have

$$\sigma(\hat{I}_{i,j})=h(\hat{I}_{i,j})$$

because $h$ is the complete history matching $\sigma$.

Then, we must also have 

$$\sigma_i(\hat{I}_{i,j})=\sigma(\hat{I}_{i,j})$$

by definition of the strategy profile.

Then, if we call the original reduced strategy for agent $i$ $\sigma'_i$, this reduced strategy must also assign this same action to our decision point because the decision actually happens in the history:

$$\sigma'_i(\hat{I}_{i,j})=\sigma_i(\hat{I}_{i,j})$$

$h$ thus perfectly matches $\sigma'_i$ on the intersection of their supports, and this is true for any agent $i$. $h$ is thus a restriction of the reduced strategy $\sigma'$. This proves existence.

For uniqueness, let us proceed ad absurdum and assume that there are two different complete histories $h$ and $h'$ that are restrictions of the reduced strategy profile $\sigma'$, and let $i$ be an agent for which $h$ and $h'$ differ on one of their decision points $\hat{I}_{i,j}$.

Since these complete histories are restrictions of $\sigma'$, and they differ on $\hat{I}_{i,j}$, then at least one of them, say $h$ without loss of generality, must assign an action to $\hat{I}_{i,j}$.

$$h(\hat{I}_{i,j})=a$$

However, since it is a restriction of $\sigma'$, we must have:

$$h(\hat{I}_{i,j})=\sigma'(\hat{I}_{i,j})=\sigma'_i(\hat{I}_{i,j})=a$$

It follows that because $h$ and $h'$ differ on this decision point,

$$h'(\hat{I}_{i,j})=\perp$$

Now, let us extend $h'$ to $h''$ by assigning

$$h''(\hat{I}_{i,j})=a$$


Now, let us pick any two decision points $\hat{I}_{k,l}$ and $\hat{I}_{m,n}$ that timelike-precede $\hat{I}_{i,j}$, and assume

$$h''(\hat{I}_{k,l})\neq \perp \wedge \gamma_{k,l}(\hat{I}_{m,n})\neq \perp$$

It follows that

$$h(\hat{I}_{k,l})\neq \perp \wedge h'(\hat{I}_{k,l})\neq \perp$$

and thus because $h$ and $h'$ are histories:

$$h(\hat{I}_{m,n})=h'(\hat{I}_{m,n})=\gamma_{k,l}(\hat{I}_{m,n})$$

Thus we must have have, by definition of $h''$:

$$h''(\hat{I}_{m,n})=\gamma_{k,l}(\hat{I}_{m,n})$$.

and $h''$ is a more complete history than $h''$, which is a contradiction.

Thus, the complete history that is a restriction of a reduced strategy profile is unique. $\square$
\end{proof}
 
 Figure \ref{figure-spacetime-example-assignments} shows the reduced strategy of Alice (A), where the second decision has been removed because it is irrelevant. It also shows the corresponding reduced strategy profile. The corresponding complete history remains unchanged.

Figure \ref{figure-reduced-strategic-form} shows the reduced strategic form of our example. The only strategy that was reduced is Alice's strategy because if her first decision is b, her second decision is irrelevant in the determination of the complete history. This removes two rows from each table. Note that, even in a reduced strategic form, the same complete histories may still appear at several places (the same applies to the reduced strategic form of games in extensive form).

\begin{figure}
\centering\includegraphics[width=\textwidth]{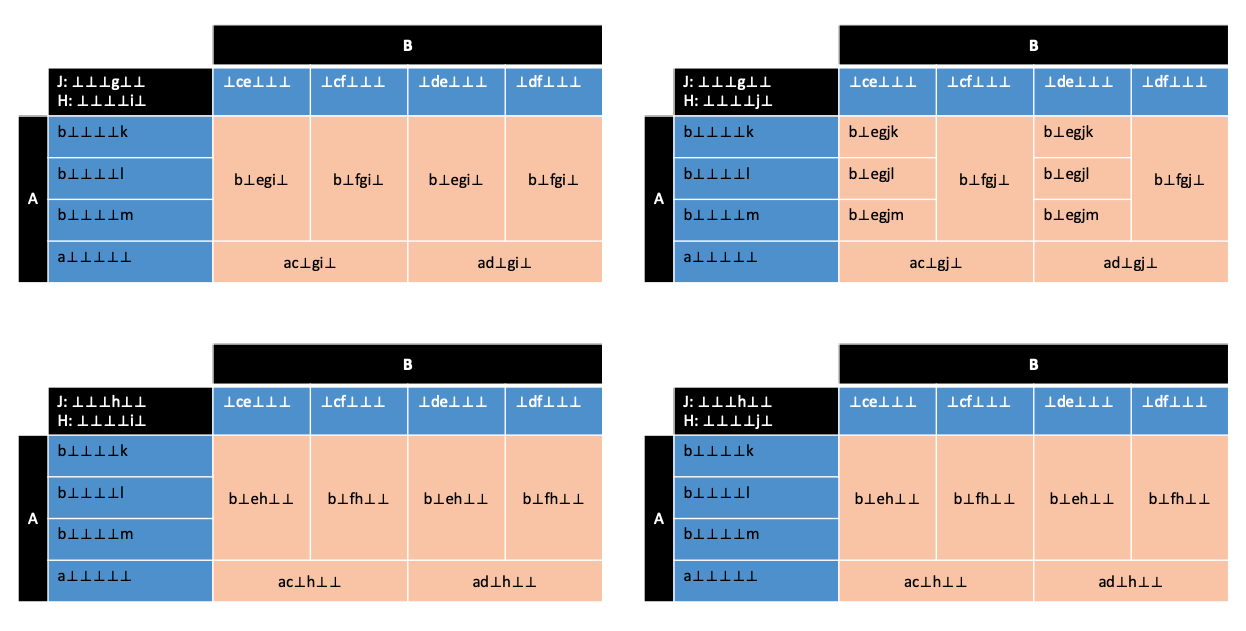}
\caption{The running example, shown in reduced strategic form. We display the complete histories rather than the payoffs, for pedagogical reasons, as the payoffs are obvious and the same as in the other figure (noting that the last three rows of each table in the (normal) strategic form are identical).}
\label{figure-reduced-strategic-form}
\end{figure}

Since we have built the reduced strategic form of a spacetime game with perfect information, many concepts from the game theory literature to spacetime games are directly applicable to spacetime games: Nash equilibria, rationalizability, individual rationality, and so on. For example, a Nash equilibria of the spacetime game is a strategy profile that is an Nash equilibrium of the strategic form of the game, and so on.

\section{Interpretation of some games in extensive form with imperfect information as spacetime games}

Given a spacetime game with perfect information, we now build an equivalent game in extensive form with incomplete information. 

\subsection{Linearization of a spacetime game}

It is possible to list all decision points in an order that is compatible with their spacetime coordinates, i.e., whenever two decision points $\hat{I}_{i,j}$ and $\hat{I}_{k,l}$ are such that $\hat{I}_{i,j}\prec \hat{I}_{k,l}$, then $\hat{I}_{i,j}$ must precede $\hat{I}_{k,l}$ in the list. Formally, we simply extend the partial order $\prec$ to a total order denoted $<$, and this choice is arbitrary.


How to build such a list is well documented in the theoretical Computer Science literature, and is known as topological ordering. Such a list is not necessarily unique because spacelike-separation gives a few degrees of freedom, but it always exists. It is known that any DAG has at least such a linear ordering, and algorithms are also known (and intuitive), and consist of a traversal of the DAG in a way that parents are always visited before their children. If it were impossible to build such a list because of a cycle in the order of events (and thus, the causal dependency graph would not be a DAG), we would be looking at a closed timelike curve, which does not exist in Minkowski spacetime\footnote{However, they may exist in general relativity, which is out of the scope of this paper, but in the scope of many quite enjoyable time-travel movies.}.

We denote the linearized list of decision points $(\hat{I}_1, \hat{I}_2, ..., \hat{I}_n)$ where $n$ is the total number of decision points. Thus, when a decision point has one index, we mean its absolute position in the ordered list (arbitrarily) selected in the former paragraph. When it has two indices, we mean, as in the rest of the paper, that the first index is the agent, and the second index its (arbitrary) index within the agent's decision points.

In our example, the list of decision points we take is $\hat{I}_{A,1}, \hat{I}_{B,1}, \hat{I}_{B,2}, \hat{I}_{J,1}, \hat{I}_{H,1}, \hat{I}_{A,2}$. Note that other lists compatible with timelike separation would be as acceptable so that the choice of linearization is really arbitrary: $\hat{I}_{A,1}, \hat{I}_{J,1},\hat{I}_{B,1}, \hat{I}_{B,2},  \hat{I}_{H,1}, \hat{I}_{A,2}$ or  $\hat{I}_{J,1}, \hat{I}_{H,1},\hat{I}_{A,1}, \hat{I}_{B,1}, \hat{I}_{B,2},  \hat{I}_{A,2}$ for example.

It is crucial to distinguish between the two orders we have introduced: a partial order with timelike-separation semantics relative to spacetime, and a total order based on the selected ordered list of decision points, the latter being an extension\footnote{We can talk about the superset of a relation because a relation between two sets is a subset of the cartesian product between these two sets.} of the former. We will denote the former $\prec$ and the latter $<$.

With a linearization of the game, it follows that complete or incomplete histories, strategies, strategy profiles, reduced strategies can also all be written as a list of (possibly missing) actions. Actually, this is what we implicitly did in Figure \ref{figure-spacetime-example-assignments}, as we used such a sequence, but we did not say at this point that this was a linearization of the DAG.

\subsection{The tree structure}

The next step is to only consider some incomplete histories, more precisely, those that are prefixes (according to the linearization) of a complete history. As we will see, they implicitly form a tree structure that will naturally yield an extensive form.

\begin{definition}[Prefix of a history]
Given a choice of linearization of a spacetime game (and reusing the same notations), we say that a history $h$ is a prefix of another history $h'$ if there exists some integer $k$ such that $h$ and $h'$ exactly match on the first $k$ decision points, and $h$ leaves all decision points beyond the $k$-th decision point unassigned.
We call $H$ the set of all incomplete histories that are a prefix of a complete history and $Z$ the set of all complete histories.
\end{definition}

Note that the empty history is also an incomplete history that matches this definition. For the purpose of building an equivalent game in extensive form with imperfect information, we only consider incomplete histories that are the prefix of a complete history, i.e., only those in $H$.

Having done this, a tree structure connecting all histories in $H\cup Z$ naturally appears, in which a parent history is always a prefix of all its children histories. Formally, this tree is the transitive reduction of the prefix relation on $H\cup Z$.

\begin{figure}
\centering\includegraphics[width=\textwidth]{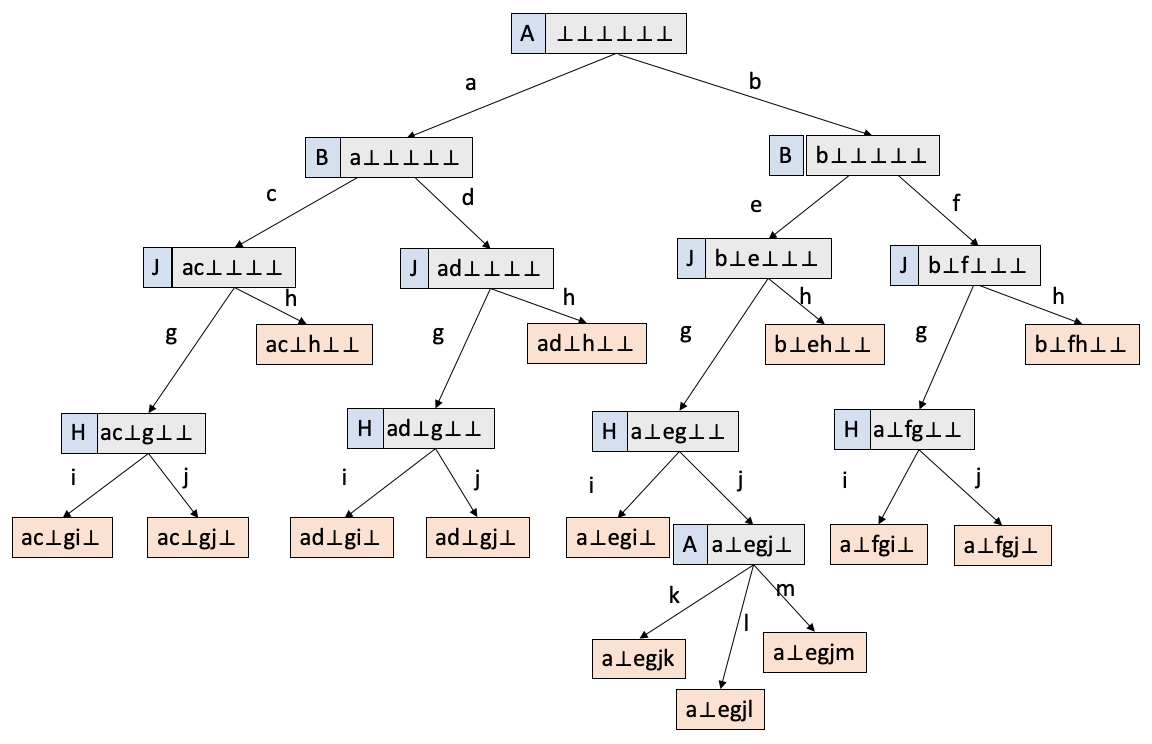}
\caption{The tree structure implicitly connecting all complete histories as well as the incomplete histories that are prefixes of a complete history, according to an arbitrary choice of linearization. Note that, if a different linearization is chosen, the tree will be different..}
\label{figure-prefix-code-tree}
\end{figure}

It directly follows from our restriction of $H$ to prefixes of complete histories that all children of a given incomplete history $h$ in $H$ differ from $h$ by assigning an action to a decision point at some position $l>k$ (the same l for all children), and there is a bijection between the children and the actions in $\hat{\chi}(\hat{I}_l)$. This is a lemma that we can prove.

\begin{lemma}
Given a choice of linearization of a spacetime game (and reusing the same notations), and an incomplete history $h$ in $H$ (according to this linearization), the children of $h$ in the transitive reduction of the prefix relation all differ from $h$ in an assignment of a decision point (the same for all children) to some choice of action taken by the agent playing at this decision point.
\end{lemma}

\begin{proof}
First, we note that the children of a node in the tree can only differ by one assignment -- otherwise, we could create a new incomplete history that is a prefix of a complete history by assigning fewer decision points, intermediate between the parent and its child, which contradicts that the tree is a transitive reduction.

Second, let us assume that two children in the tree differ from their parent by assigning an action to different decision points $\hat{I}_m < \hat{I}_n$ (according to the linearization). Let us call these two incomplete histories $h$ and $h'$ in that order. Let $h''$ be a complete history that extends $h'$, i.e., $h'$ is a prefix of $h'$. It follows that

$$h'(\hat{I}_m)=h''(\hat{I}_m)=\perp$$

Let $h'''$ be a complete history that extends $h$, i.e., $h$ is a prefix of $h'''$. It follows that

$$h'''(\hat{I}_m)=h(\hat{I}_m)\neq\perp$$

And thus, by definition of a complete history, if $\gamma$ denotes the contingency coordinate mapping:

$$\forall \hat{I}_p, \gamma_m(\hat{I}_p)\neq \perp \implies h'''(\hat{I}_p)=\gamma_m(\hat{I}_p)$$.

But then, since $h''$ and $h'''$ exactly match at all decision points that precede $\hat{I}_m$ (and all decision points such that $ \gamma_m(\hat{I}_p)\neq \perp$ precede it by definition of $\gamma$):

$$\forall \hat{I}_p, \gamma_m(\hat{I}_p)\neq \perp \implies h''(\hat{I}_p)=\gamma_m(\hat{I}_p)$$.

Since $h''$ is a complete history, however, it implies that

$$h''(\hat{I}_m)\neq\perp$$

which contradicts what we established above.

As a consequence, it cannot be that two children of the same parent, in the transitive reduction of the prefix relation, different from the parent at different positions.$\square$
\end{proof}

We can now use this lemma to define two new functions. First, we define a function $\rho$ that maps any incomplete history $h$ in $H$ to the agent deciding at $\hat{I}_l$, where $\hat{I}_l$ is the decision point defined above for $h$. Second, we define the function $\sigma$ mapping an incomplete history $h$ with decision point $\hat{I}_l$ as defined above, and an action $a\in \hat{\chi}(\hat{I}_l)$ to the (incomplete or complete) history in $H\cup Z$ obtained by assigning $a$ to $\hat{I}_l$.

Also, we can use this lemma to partition H according to the decision point $\hat{I}_l$ at which their children differ in their assignment: all incomplete histories in H that share this same decision point $\hat{I}_l$ can be grouped in the same partition, and this partition can be canonically identified with $\hat{I}_l$ -- we denote this partition $I_l$, without the hat, and $I$ the overall partitioning that it induces on $H$. Likewise, we denote $\chi$ (without hat) the corresponding functions on $H$ such that $\chi(h)=\hat\chi(\hat{I}(h))$.

$\sigma$ is injective by construction because the previous incomplete history and action can be reconstructed straightforwardly by looking up and unassigning this action. The unique root of the suchly obtained tree is the empty history.

For example, in our running example of spacetime game, $$\sigma((a, c, \perp, \perp, \perp, \perp), g) = (a, c, \perp, g, \perp, \perp)$$

\subsection{Construction of the game with imperfect information}

We now have all the components necessary to build an equivalent game in extensive for with imperfect information.
\begin{theorem}
Given a spacetime game $$\Gamma=(N, A, \hat{I}, \hat{\chi}, A, u)$$
the tuple $$\hat\Gamma=(N, A, H, Z, \chi, \rho, \sigma, u, I)$$ where $H, Z, \rho, \sigma, \chi$ and $I$ are constructed with the procedure shown previously, is  a game in extensive form with imperfect information. Furthermore, it has the same strategic form as the spacetime game.
\end{theorem}

\begin{proof}
The sets of agents, actions, nodes, outcomes suchly constructed do not require any proof, as they are arbitrary. $\chi$ is by construction consistent with the definition of a game with imperfect information because it maps each node to an action. $\rho$ is also by construction consistent because it maps each node to an agent.

$\sigma$ is by construction mapping pairs of node and action to either a node (incomplete history) or an outcome (complete history). It is injective by construction because the antecedent node and action can be directly reconstructed from the image by just splitting it. Also, $\sigma$ defines a graph on nodes and outcomes that has a unique connected component, forming a tree the root of which is the empty history.

$u$ is by construction consistent because it maps each outcome and agent to a real number.
Finally, $I$ is a partition of $H$ by construction. It is compatible with $\rho$ by construction of $\rho$. Likewise, it is compatible with $\chi$ by construction of $\chi$.

The fact that both games have the same strategic form follows directly from the fact that the strategy spaces are obtained, in both cases, by taking the cartesian products of each agent's possible choice of action at all the decision points or information sets at which they play.
$\square$
\end{proof}

\begin{figure}
\centering\includegraphics[width=\textwidth]{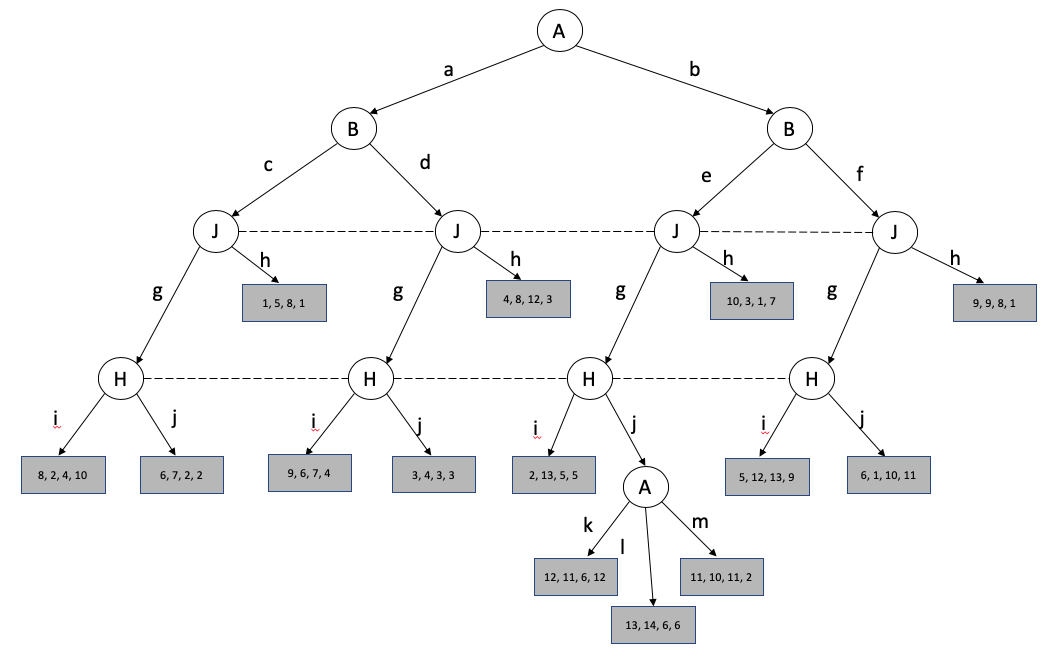}
\caption{The game associated with the example spacetime configuration with six decision points and four agents. Each consistent complete history corresponds to an outcome (gray node) and each consistent incomplete history corresponds to a choice node (white node). The letters are the initials of the agents making decisions. The actions are indicated on the edges. Information sets are shown with dashed lines. The outcomes are numbered arbitrarily. The preferences of the agents are not shown, but would appear as tuples of four numbers at each outcome with ordinal semantics.}
\label{figure-imperfect-information-equiv}
\end{figure}

Figure \ref{figure-imperfect-information-equiv} shows the game in extensive form and with imperfect information that is obtained from the spacetime game described in Figures \ref{figure-spacetime-example}, \ref{figure-spacetime-example-timelike-dependence}, \ref{figure-spacetime-example-assignments}, and \ref{figure-strategic-form}.

A game obtained via this construction scheme thus has a natural interpretation in which the players are agents located in Minkowski spacetime making decisions. The information sets are interpreted as the situations in which decisions are spacelike-separated and thus no signal can be sent between two agents.

Conversely, some (but not all) games in extensive form and with imperfect information can be interpreted as a spacetime game with perfect information.

\begin{definition}
A game in extensive form and with imperfect information $\Gamma$ is said to be interpretable as a spacetime game if there exists a spacetime game $$\Gamma'=(N, A, \hat{I}, \hat{\chi}, A, u)$$ such that $$\hat\Gamma'=(N, A, H, Z, \chi, \rho, \sigma, u, I)$$ where $H, Z, \chi, \rho, \sigma$ and $I$ are constructed with the procedure shown previously is isomorphic to $\Gamma$. By isomorphic, we mean that the two games $\Gamma$ and $\hat\Gamma'$ are exactly the same, up to a simple bijection on the sets.
\end{definition}

Not all games with imperfect information can be obtained in this way. Further work includes exactly characterizing this subclass of games in extensive form and with imperfect information that is interpretable as spacetime games. It is clear that, for example, games with imperfect recall cannot be interpreted as spacetime games because spacetime games always correspond to games with imperfect information and perfect recall.

But perfect recall alone still is not enough to characterize games in extensive form and with imperfect information that are interpretable as spacetime games. Figure \ref{figure-counterexample} provides a counterexample of game in extensive form with imperfect information and perfect recall that cannot be interpreted as an equivalent spacetime game with perfect information.

\begin{figure}
\centering\includegraphics[width=\textwidth]{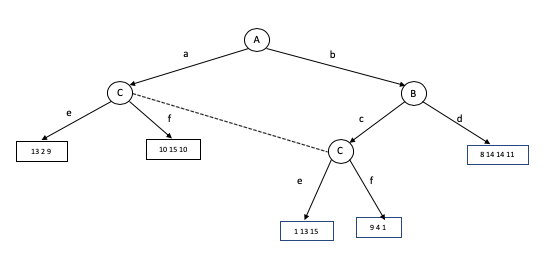}
\caption{A game in extensive form with imperfect information that cannot be interpreted as a spacetime game with perfect information. Indeed, C knows B's decision and B knows A's decision, but C does not know A's decision, which would be incompatible with the transitivity of timelike-precedence in a spacetime game.}
\label{figure-counterexample}
\end{figure}

We could also consider that a game in extensive form and with imperfect information is indirectly interpretable as a spacetime game if it is equivalent to a game that can be interpreted as such. By equivalent, we mean that they have the same reduced strategic form.

However, it would follow that any game is indirectly interpretable as a spacetime game: indeed, we can simply build a spacetime game out of its reduced strategic form, interpreted as a game in normal form and thus, spacelike-separation; this is thus not very helpful.

 \citet{Thompson1997} characterized the games that share the same reduced strategic forms by showing that they can be obtained from each other by a chain of transformations of only four kinds. These transformations are worth discussing because some can be interpreted in terms of spacetime. They are also discussed by \citet{Rubinstein1994}.

The first kind of transformation is called \emph{principle of interchange of moves}. This kind of transformation corresponds to a different choice of linearization of the spacetime game: two spacelike-separated decisions can appear in any order in the game in extensive form with imperfect information so that the same spacetime game can correspond to different games in extensive form with imperfect information that differ in this way.

The second kind of transformation is called \emph{principle of coalescing of moves}. This corresponds to merging or splitting one decision by an agent into several decisions. Doing so in the game in extensive form amounts to doing the same in the equivalent spacetime game.

The third kind of transformation is called \emph{principle of addition of a superfluous move}. This kind of transformation corresponds to the addition of decision points that can never be actually reached, or to the planning ahead of strategies (i.e., plan decisions that may not actually go into effect). Since in our spacetime game, we prune all decisions that cannot be reached, our transformation can only produce games in extensive form with no superfluous moves, in particular, all outcomes are always distinct.

The fourth kind of transformation is called \emph{principle of inflation-deflation}. This kind of transformation corresponds to the introduction (or removal) of the constraint that players have perfect recall. If we require that the same agent cannot make a decision at two spacelike-separated (including colocated) decision points, then the corresponding game with imperfect information has perfect recall. This requirement is natural and corresponds to the fact that the timeline\footnote{the set of all spacetime coordinates ever occupied by an agent} of an agent must follow a timelike curve, that no observer in Minkowski spacetime can see the same agent at two spacelike-separated positions\footnote{This could, however, occur in general relativity in the presence of closed timelike curves, which is out of the scope of this paper.}, and that an agent can only make one decision at any time\footnote{We mean here their proper time.}.

\section{Another example: the EPR experiment}

Let us illustrate the contribution of this paper with another example taken from quantum foundations: the EPR experiment, where the series of decisions and measurements are interpreted as a game.

Figure \ref{figure-epr-spacetime} shows a visual representation of the configuration of this spacetime game. Alice and Bob are two physicists located far away from each other, by which we mean spacelike-separated. For example, Alice could be on Earth or Bob on Mars\footnote{In practice, experiments are carried out on Earth, but with very precise timing in order to ensure spacelike-separation. Physicists consider issues due to the lack of enforcement of spacelike separation in the protocol of experiments seeking to confirm that Nature breaks Bell inequalities to be ``loopholes''. See for example the Big Bell Test paper \citep{BigBellTest}}.

Alice and Bob each have a particle in their laboratory and can pick what they observe. The particles may be entangled. For simplicity, we will say that they can either observe the color (c) or the form (f) of this particle\footnote{In the real experiment, this could be, for example, the position or the momentum of a particle, or the choice of a horizontal or diagonal basis to measure a spin.}. Then, if Alice picked c, Ulysses can pick green (g) or blue (b); or if Alice picked f, Ulysses can pick round (r) or square (s). The same on Bob's side: Valentina can pick green (g) or blue (b) if Bob picked c, and she can pick round (r) or square (s) if Bob picked f. In total, these are six decision points.

Under this configuration, a domain of intense research in quantum foundations is to look at the possible joint probabilities over the four agents (the choices of Ulysses and Valentina, conditioned on the choices of Alice and Bob) and derive inequalities that this joint distribution must fulfill under certain hypotheses (e.g., locality, realism, free choice...). For more on this, see a tutorial by \citet{Colbeck2017}.

\begin{figure}
\centering\includegraphics[width=\textwidth]{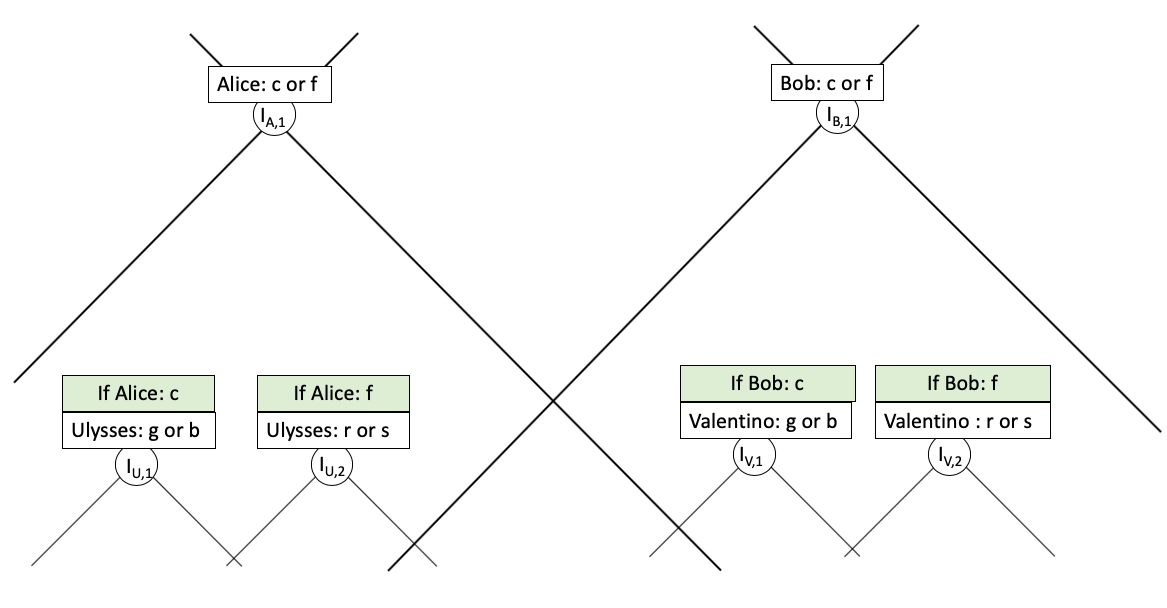}
\caption{The EPR experiment, where Alice and Bob are experimenters and Ulysses and Valentina represent the universe deciding on a measurement outcome.}
\label{figure-epr-spacetime}
\end{figure}

Figure \ref{figure-epr-dependence} shows the corresponding causal dependency graph (a DAG), where we can see that there are two connected components: one for Alice's laboratory and one for Bob's laboratory. In this example, it turns out that the actual precedence relation coincides exactly with timelike precedence. Figure \ref{figure-epr-contingency-coordinates} show the contingency coordinates corresponding to the configuration that we described, specifically, Ulysses and Valentina get to make decisions depending on what Alice and Bob previously decided. For these contingency coordinates, we arbitrarily picked a linearization of the DAG of the game, which corresponds to the order of the rows and columns on the corresponding matrix. Note that, if the contingency coordinates are displayed according to a linearization of the game, actual assignments (i.e., not $\perp$) always appear in a lower triangular matrix only.

\begin{figure}
\centering\includegraphics[width=0.5\textwidth]{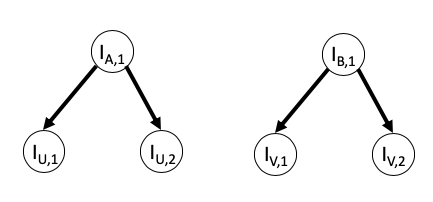}
\caption{The EPR experiment's causal dependency graph. In this specific game, it turns out that the actual precedence relation is identical to the timelike precedence relation so that this is also the graph of the actual precedence relation of the game.}
\label{figure-epr-dependence}
\end{figure}

\begin{figure}
\centering\includegraphics[width=0.7\textwidth]{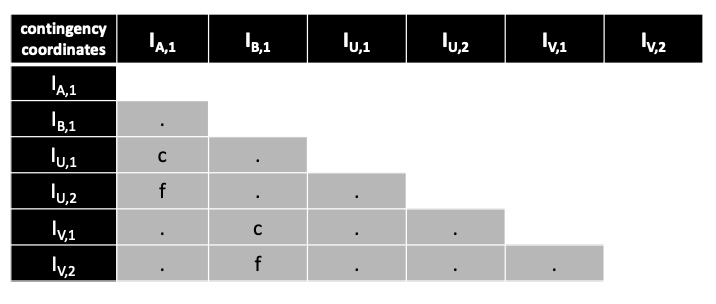}
\caption{The contingency coordinates of the EPR game.}
\label{figure-epr-contingency-coordinates}
\end{figure}

It is furthermore straightforward to assign payoffs (arbitrarily for the purpose of this paper) to each possible outcome of the experiment and turn it into a game where each agent makes their decision in order to maximize their utility.

Let us now build the histories. First, the complete histories and their associated distributions (in this order: Alice, Bob, Ulysses, Valentina) are

\begin{itemize}
\item $c, g, \perp, c, g, \perp$, assigned with payoff distribution (13, 2, 9, 1)
\item $c, g, \perp, c, b, \perp$, assigned with payoff distribution (10, 15, 10, 12)
\item $c, b, \perp, c, g, \perp$, assigned with payoff distribution (14, 16, 11, 3)
\item $c, b, \perp, c, b, \perp$, assigned with payoff distribution (16, 9, 8, 2)
\item $c, g, \perp, f, \perp, r$, assigned with payoff distribution (15, 3, 6, 13)
\item $c, g, \perp, f, \perp, s$, assigned with payoff distribution (11, 6, 3, 4)
\item $c, b, \perp, f, \perp, r$, assigned with payoff distribution (6, 1, 16, 5)
\item $c, b, \perp, f, \perp, s$, assigned with payoff distribution (12, 7, 7, 6)
\item $f, \perp, r, c, b, \perp$, assigned with payoff distribution (1, 13, 15, 7)
\item $f, \perp, r, c, g, \perp$, assigned with payoff distribution (9, 4, 1, 14)
\item $f, \perp, s, c, b, \perp$, assigned with payoff distribution (5, 11, 4, 15)
\item $f, \perp, s, c, g, \perp$, assigned with payoff distribution (2, 8, 12, 12)
\item $f, \perp, r, f, \perp, r$, assigned with payoff distribution (8, 14, 14, 11)
\item $f, \perp, r, f, \perp, s$, assigned with payoff distribution (4, 5, 5, 10)
\item $f, \perp, s, f, \perp, r$, assigned with payoff distribution (3, 12, 13, 8)
\item $f, \perp, s, f, \perp, s$, assigned with payoff distribution (7, 10, 2, 9)
\end{itemize}

And the incomplete histories that are prefixes of complete histories according to our choice of linearization are

\begin{itemize}
\item $\perp, \perp, \perp, \perp, \perp$
\item $c, \perp, \perp, \perp, \perp, \perp$
\item $c, b, \perp, \perp, \perp, \perp$
\item $c, b, \perp, c, \perp, \perp$
\item $c, b, \perp, f, \perp, \perp$
\item $c, g, \perp, \perp, \perp, \perp$
\item $c, g, \perp, c, \perp, \perp$
\item $c, g, \perp, f, \perp, \perp$
\item $f, \perp, \perp, \perp, \perp, \perp$
\item $f, \perp, r, \perp, \perp, \perp$
\item $f, \perp, r, c, \perp, \perp$
\item $f, \perp, r, f, \perp, \perp$
\item $f, \perp, s, \perp, \perp, \perp$
\item $f, \perp, s, c, \perp, \perp$
\item $f, \perp, s, f, \perp, \perp$
\end{itemize}

If we apply the conversion algorithm described in this paper, we obtain the game in extensive form and with imperfect information shown in Figure \ref{figure-epr-game}. As can be seen, this game has six information sets, one for each decision points in the EPR spacetime game: one for Alice, one for Bob, and two for each of Ulysses and Valentina. Each of the sixteen outcomes corresponds to a consistent complete history, and each of the fifteen nodes corresponds to a consistent incomplete history.
 
\begin{figure}
\centering\includegraphics[width=\textwidth]{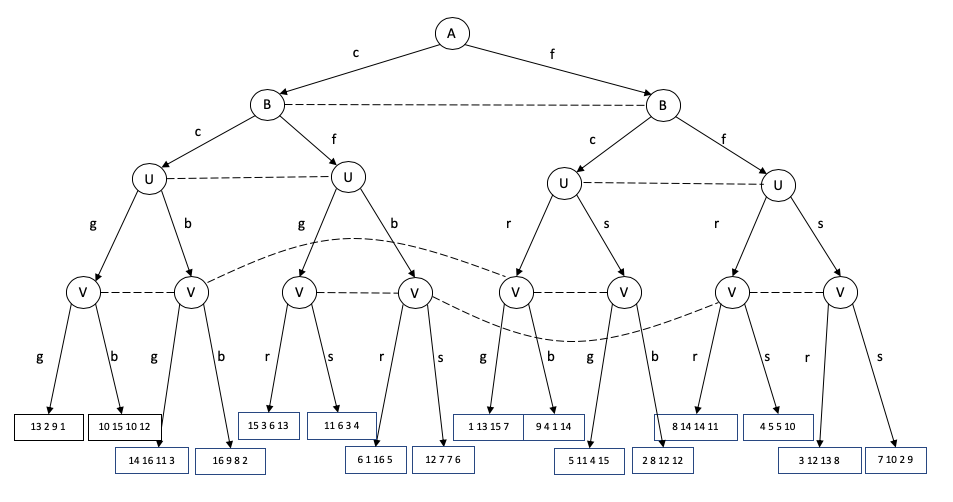}
\caption{The EPR experiment, where Alice and Bob are experiments, and Ulysses and Valentina represent the universe deciding on a measurement outcome.}
\label{figure-epr-game}
\end{figure}

This game can them be solved for any solution concept that makes sense for a game in extensive form with imperfect information.

\section{Conclusion}

In this paper, we introduced a new class of games played in Minkowski spacetime. They provide a least common denominator that unifies strategic games and dynamic games with perfect information. This framework is more fine-grained than the provably strictly wider class of all games in extensive form with imperfect information. Furthermore, this class of games is consistent with the theory of special relativity, according to which there cannot exist any such thing as games with (absolutely) simultaneous moves. Information sets are interpreted as decision points placed in spacetime, and imperfect information is interpreted as the spacelike separation of some decisions that cannot be known to each other. Games in normal form correspond to spacelike-separated decisions, and games in extensive form with perfect information to timelike-separated decisions. Equilibrium concepts defined for games in extensive form and with imperfect information can thus directly be used to describe rational choice in Minkowski spacetime by using the corresponding mapping. Any equilibrium concept, such as Nash equilibria, can naturally and directly be applied to spacetime games.


\bibliographystyle{spbasic}      

\end{document}